\documentclass[runningheads]{llncs}
\usepackage{hyperref}
\usepackage[T1]{fontenc}
\usepackage[utf8]{inputenc}
\usepackage{xspace}
\usepackage{graphicx}
\usepackage{listings}
\usepackage{subcaption}
\usepackage{xcolor}
\usepackage{xspace}
\usepackage{float}
\usepackage{lstautogobble}
\usepackage{color}
\usepackage{amsmath}
\usepackage{amssymb}
\usepackage{stmaryrd}
\usepackage{proof}
\usepackage{multicol}
\usepackage{mathpartir}
\usepackage{todonotes}
\usepackage{cleveref}
\usepackage{mathpartir}
\usepackage{todonotes}
\usepackage{wrapfig}
\usepackage{mybnf}
\usepackage{mathpartir}
\usepackage{amsmath, amsfonts, amssymb}
\usepackage{stmaryrd}
\usepackage{todonotes}
\usepackage{orcidlink}

\definecolor{bluekeywords}{rgb}{0.13, 0.13, 1}
\definecolor{greentypes}{rgb}{0, 0.5, 0}
\definecolor{orangecomments}{rgb}{1, 0.5, 0.1}
\definecolor{redstrings}{RGB}{171, 114, 2}
\definecolor{graynumbers}{rgb}{0.5, 0.5, 0.5}
\definecolor{goldcomments}{rgb}{0.6, 0.4, 0.08}

\lstdefinelanguage{Lola}{
  keywords=[0]{input, output, trigger, constant, import, spawn, eval, close, with, when},
  moredelim=**[is][\color{greentypes}@]{@}{@},
  keywordstyle=[0]\bfseries\color{bluekeywords},
  keywords=[1]{if, then, else, aggregate, defaults, offset, by, or, to, sin, cos, abs, hold, over, using, over_instances},
  keywords=[2]{Variable, String, Int, Int64, UInt, UInt64, Bool, Float32, Float64, Float},
  keywordstyle=[2]\color{greentypes},
  sensitive=false,
  comment=[l]{//},
  morecomment=[s]{/*}{*/},
  morestring=[b]',
  morestring=[b]",
  literate={\\@}{@}1
}
\newcommand{\rtlola}[0]{RTLola\xspace}
\newcommand{\con}{\sqsubseteq}
\newcommand{\streamName}[1]{{\scriptsize\texttt{#1}}}
\newcommand{\streamKeyword}[1]{{\color{bluekeywords}\texttt{#1}}}

\newcommand{\outputs}[0]{\texttt{Output}}

\newcommand{\exprs}[0]{\texttt{Expr}}
\newcommand{\ids}[0]{\texttt{ID}}
\newcommand{\idsI}[0]{\ids^\uparrow}
\newcommand{\idsO}[0]{\ids^\downarrow}

\newcommand{\pacingAny}[0]{\texttt{Any}}

\newcommand\syn[1]{\texttt{\color{brown}{#1}}}
\newcommand{\constant}[1]{\syn{\texttt{Constant}(}#1\syn{)}}
\newcommand{\syncAccess}[2]{\syn{\texttt{Sync}(}#1, #2\syn{)}}
\newcommand{\parameterAccess}[1]{\syn{\texttt{Parameter}(}#1\syn{)}}
\newcommand{\offsetAccess}[3]{\syn{\texttt{Offset}(}#1, #2, #3\syn{)}}
\newcommand{\holdAccess}[2]{\syn{\texttt{Hold}(}#1, #2\syn{)}}

\newcommand{\functionExp}[2]{\syn{\texttt{Function}(}#1, #2\syn{)}}

\newcommand{\default}[2]{\syn{\texttt{Default}(}#1, #2\syn{)}}
\newcommand{\aggr}[4]{\syn{\texttt{Aggregate}(}#1, #2, #3, #4\syn{)}}
\newcommand{\tuple}[1]{\syn{\texttt{Tuple}(}#1\syn{)}}
\newcommand{\localFreq}[1]{\syn{\texttt{Local}(}#1\syn{)}}
\newcommand{\globalFreq}[1]{\syn{\texttt{Global}(}#1\syn{)}}
\newcommand{\eventAc}[1]{\syn{\texttt{Event}(}#1\syn{)}}

\newcommand{\sInput}[1]{\syn{\texttt{Input}(}#1\syn{)}}
\newcommand{\sOutput}[5]{\syn{\texttt{Output}(}#1, #2, #3, #4, #5\syn{)}}
\newcommand{\sSpawn}[3]{\syn{\texttt{Spawn}(}#1, #2, #3\syn{)}}
\newcommand{\sEval}[3]{\syn{\texttt{Eval}(}#1, #2, #3\syn{)}}
\newcommand{\sClose}[2]{\syn{\texttt{Close}(}#1, #2\syn{)}}

\newcommand{\VT}{\textit{VT}}
\newcommand{\vtPrim}{\textit{VT}_\textit{prim}}
\newcommand{\vtCon}{\con_\textit{VT}}
\newcommand{\vt}{v}

\newcommand{\PT}{\textit{PT}}
\newcommand{\ac}{\textit{AC}}
\newcommand{\ptCon}{\con_\textit{PT}}
\newcommand{\pt}{\tau}
\newcommand{\ptt}{\mathcal{T}}
\newcommand{\lf}[1]{\mathit{LocalFreq}(#1)}
\newcommand{\gf}[1]{\mathit{GlobalFreq}(#1)}
\newcommand{\periodic}{\mathit{Periodic}}
\newcommand{\eventT}[1]{\mathit{Event}(#1)}

\newcommand{\ST}{\textit{ST}}
\newcommand{\stCon}{\con_\textit{ST}}
\newcommand{\st}{\pi}
\newcommand{\stt}{\Pi}

\newcommand{\optional}[1]{\textit{Option}(#1)}
\newcommand{\bool}{\textit{Bool}}
\newcommand{\types}[1]{\models_{#1}}

\newcommand{\paramEq}[2]{\cong_{#2}^{#1}}
\newcommand{\parameterTypes}[1]{\rho(#1)}
\newcommand{\streamTypes}[1]{\alpha(#1)}

\newcommand{\semEval}[1]{\llbracket #1 \rrbracket}
\newcommand{\srin}[0]{\idsI}
\newcommand{\srout}[0]{\idsO}
\newcommand{\sr}[0]{\srin \uplus \srout}

\newcommand{\StreamMap}[0]{\mathit{StreamMap}}

\newcommand{\instanceId}[0]{\mathit{InstanceID}}
\newcommand{\Stream}[0]{\mathit{Stream}}
\newcommand{\prefix}[0]{\mathit{Prefix}}
\newcommand{\stime}[0]{\mathit{Time}}
\newcommand{\svalue}[0]{\mathbb{V}_{\bot}}

\newcommand{\pacing}[3]{\llbracket#2\rrbracket^{#1}_{\world, #3}}
\newcommand{\pacingAC}[2]{\llbracket#2\rrbracket^{#1}_{\world}}

\newcommand{\world}[0]{\omega}
\newcommand{\World}[0]{\mathbb{W}}
\newcommand{\evalToAt}[4]{{#3}\mathbin{\Downarrow^{#1}_{\world,#2}}{#4}}
\newcommand{\evalToWith}[3]{\evalToAt{t}{#1}{#2}{#3}}
\newcommand{\evalTo}[2]{\evalToWith{\mathit{inst}}{#1}{#2}}
\newcommand{\lhs}[0]{\mathit{lhs}}
\newcommand{\rhs}[0]{\mathit{rhs}}
\newcommand{\val}[0]{\mathit{val}}
\newcommand{\expr}[0]{\mathit{expr}}
\newcommand{\worldaccone}[1]{\world(#1)}
\newcommand{\worldacctwo}[2]{\world(#1)(#2)}
\newcommand{\worldaccthree}[3]{\world(#1)(#2)(#3)}

\newcommand{\pac}[0]{\mathit{pac}}

\newcommand{\stimeMap}[0]{\mathit{TimeMap}}
\newcommand{\tstart}[0]{\mathit{start}_t}
\newcommand{\tend}[0]{\mathit{end}_t}
\newcommand{\instance}[0]{\mathit{inst}}
\newcommand{\instanceacc}[1]{\instance(#1)}
\newcommand{\alive}[0]{\mathit{AliveSince}}
\newcommand{\sid}[0]{\mathit{sid}}
\newcommand{\instid}[0]{\mathit{iid}}
\newcommand{\spawnPac}[1]{\mathit{SpawnPac}(#1)}
\newcommand{\spawnWhen}[1]{\mathit{SpawnWhen}(#1)}
\newcommand{\spawnWith}[1]{\mathit{SpawnWith}(#1)}
\newcommand{\evalPac}[1]{\mathit{EvalPac}(#1)}
\newcommand{\evalWhen}[1]{\mathit{EvalWhen}(#1)}
\newcommand{\closePac}[1]{\mathit{ClosePac}(#1)}
\newcommand{\closeWhen}[1]{\mathit{CloseWhen}(#1)}
\newcommand{\isSpawned}[0]{\mathit{IsSpawned}}
\newcommand{\isClosed}[0]{\mathit{IsClosed}}
\newcommand{\window}[0]{\mathit{Window}}

\renewcommand{\lstinline}[1]{\texttt{#1}}
\begin{document}
%
\title{Type-safe Monitoring of Parameterized Streams}

%
%
\author{Jan Baumeister\inst{1}\orcidlink{0000-0002-8891-7483} \and
  Bernd Finkbeiner\inst{1}\orcidlink{0000-0002-4280-8441} \and
  Florian Kohn\inst{1}\orcidlink{0000-0001-9672-2398}}
\authorrunning{Baumeister, Finkbeiner, and Kohn}
\institute{}
%
\institute{CISPA Helmholtz Center for Information Security\newline66123 Saarbrücken, Germany\newline
  \email{\{jan.baumeister, finkbeiner, florian.kohn\}@cispa.de}}

\maketitle              
\begin{abstract}
  Stream-based monitoring is a real-time safety assurance mechanism for complex cyber-physical systems such as unmanned aerial vehicles.
  The monitor aggregates streams of input data from sensors and other sources to give real-time statistics and assessments of the system's health.
  Since the monitor is a safety-critical component, it is mandatory to ensure the absence of runtime errors in the monitor.
  Providing such guarantees is particularly challenging when the monitor must handle unbounded data domains, like an unlimited number of airspace participants, requiring the use of dynamic data structures.
  This paper provides a type-safe integration of parameterized streams into the stream-based monitoring framework RTLola.
  Parameterized streams generalize individual streams to sets of an unbounded number of stream instances and provide a systematic mechanism for memory management.
  We show that the absence of runtime errors is, in general, undecidable but can be effectively ensured with a refinement type system that guarantees all memory references are either successful or backed by a default value.
  We report on the performance of the type analysis on example specifications from a range of benchmarks, including specifications from the monitoring of autonomous aircraft.


\end{abstract}
\section{Introduction}
\label{sec:intro}
Stream-based monitoring is a runtime verification approach for complex cyber-physical systems (CPS) such as unmanned aerial vehicles.
Unlike monitoring approaches using logics like LTL~\cite{bauer2011runtime}, stream-based languages can reason over rich data such as integers or strings.
A runtime monitor, or monitor, collects data through input streams, which are then combined and aggregated in output streams to give real-time assessments of the system's health.
Before generating a monitor from the specification, the specification is statically analyzed to establish key properties like a static memory bound or the absence of runtime errors.
%
%

This approach has already been applied successfully in real-world applications.
One prominent example is the stream-based specification language \rtlola that has been applied in monitoring unmanned aircrafts~\cite{volostream} or monitoring the emission of cars \cite{DBLP:conf/tacas/BiewerFHKSS21}.
As systems grow more complex, monitors must handle unbounded data domains.
For example, when monitoring a UAV's airspace, there can be no limit on the number of nearby vehicles.
A simple solution is to extend stream data types to include structures like Maps or Lists.
However, the specifier must ensure values are present in these structures, or the monitor could fail at runtime.
%
%
%
%
%
%

A competing approach to handle unbounded domains directly in the specification are parameterized streams as introduced by Faymonville et.al.~\cite{Lola2.0}.
Parameterized streams generalize streams to unbounded sets of instances, which can be created and closed at runtime to manage memory.
When stream instances are created, evaluated, and closed is expressed in the logic, allowing to statically guarantee the absence of runtime errors stemming from invalid memory accesses.
Yet, current approaches cannot express asynchronous or real-time properties, which are essential in monitoring cyber-physical systems.
%
%
%

This paper introduces parameterized streams to the real-time stream-based specification language \rtlola, addressing the challenge of dynamically creating stream instances at runtime.
In earlier versions of \rtlola, the existence of a stream's value could be determined statically using annotated types, such as frequencies or events.
However, with the introduction of parameterized streams, the existence of a stream's value depends on runtime data.
Consequently, previous static analysis methods are no longer adequate.

We propose a new refinement type system for \rtlola that argues over the time points in which streams evaluate.
There, statically known time points defined by, for example, a frequency annotation, can be refined using boolean stream expressions that depend on runtime values.
We demonstrate that these types capture the dynamic behavior of parameterized streams in \rtlola and show that well-typed specifications can be monitored without runtime errors.

For this, we give a semantics for parameterized streams in \rtlola and show that guaranteeing the absence of runtime errors is undecidable.
We extend the definition of well-formedness to prevent nondeterministic behavior of the monitor.
We implement the approach in the \rtlola-framework \cite{volostream} using an over-approximation of the type system and evaluate its runtime using specifications from the aerospace domain and synthetic specifications to validate its scalability.

The rest of \Cref{sec:intro} demonstrates the benefits of parameterized streams and discusses related work.
In \Cref{sec:syntax}, we define the syntax before providing the semantics in \Cref{sec:semantics}.
\Cref{sec:types} defines the new type system and further explains its limitations.
\Cref{sec:welldefined} expands the definition of wellformedness, while \Cref{sec:evaluation} evaluates the performance of our type-checking approach and \Cref{sec:conclusion} concludes this paper.
%
%
%
%
%
%
%


\subsection{Motivating Example}
\label{sec:intro:example}
This example presents the usage of the real-time features of \rtlola together with the dynamic creation of stream instances to monitor the airspace around an autonomous drone. 
It is a simplified version of the example Baumeister et al. presented in \cite{DBLP:conf/fm/BaumeisterFKS24}.
The monitor, specified using the \rtlola specification language, tracks the distance to an a priori unbounded number of airspace participants.
It notifies the aircraft operator in case an airspace participant, called an intruder, approaches the aircraft and gets critically close.
%
	\begin{lstlisting}
input intruder_id: UInt
input distance: Float

output distance_per(id)
    spawn with intruder_id
    eval when id = intruder_id with distance
    close when id = intruder_id && distance > 10.0

output avg_distance(id)
    spawn with intruder_id
    eval @1Hz@ with distance_per(id).aggregate(over: 1s, using: avg).defaults(to: 0.0)  
    close when id = intruder_id && distance > 10.0

trigger(id)
    spawn when distance < 5.0 with intruder_id
    eval @1Hz@ when avg_distance(id).last(or: 0.0) > avg_distance(id) with "Intruder"
    close when id = intruder_id && distance > 10.0
\end{lstlisting}
We assume the aircraft is equipped with a sensor that measures the distance to an intruder.
The example includes two input streams, capturing the aircraft's observations.
The first input stream, \streamName{intruder\_id}, represents the intruder's identifier, and the second stream, \streamName{distance}, represents the distance to that intruder.

The first output stream, \streamName{distance\_per}, groups distance measurements by intruder using a single parameter, \streamName{id}, indicated by parentheses after the stream name. 
The \streamKeyword{spawn} clause assigns this parameter using the \streamName{intruder\_id} input stream, creating one stream instance per intruder. 
Each instance {\color{bluekeywords}eval}uates to the current distance \streamKeyword{when} the intruder matches its id and is removed from memory when the distance exceeds a specified threshold.

To reduce sensor fluctuations, the \streamName{avg\_distance} stream computes the average distance at a 1Hz frequency.
The final trigger tests whether an intruder is approaching by comparing the current and previous values of the \streamName{avg\_distance} stream once an intruder is closer than a threshold.
Monitoring this example can actually result in runtime errors, as the trigger's periodicity and the \streamName{avg\_distance} stream may not be synchronized. 
This desynchronization can cause the trigger to access invalid or unavailable values of the \streamName{avg\_distance} stream, leading to memory errors. 
The following sections show how the proposed type system effectively detects these errors, ensuring safe and reliable monitoring.

\subsection{Related Work}
\label{sec:related}
Runtime verification has been applied to a wide range of systems and properties \cite{DBLP:conf/cav/JungesTS20,DBLP:conf/cav/HenzingerKKM23}.
Properties can, for example, be specified using temporal logics like LTL \cite{bauer2011runtime} or MTL \cite{DBLP:conf/cav/BasinKZ17}.
Stream-based monitoring approaches \cite{d2005lola,DBLP:conf/rv/KallwiesLSSTW22,DBLP:conf/fm/GorostiagaS21} use specification languages that handle non-binary data and are inspired by synchronous languages \cite{DBLP:journals/pieee/HalbwachsCRP91,DBLP:conf/concur/BerryC84}.
One way these languages can handle dynamic memory is by introducing new data types.
HLola~\cite{DBLP:conf/tacas/GorostiagaS21}, for example, embeds data types defined in Haskell in Lola\cite{d2005lola}.
With this approach, the responsibility of managing dynamic memory is given to the specifier.
An alternative to this approach are parametrized streams~\cite{Lola2.0}.
Dynamic memory is directly embedded in the specification language by describing stream templates.
This approach was later embedded into other stream-based languages such as HStriver~\cite{DBLP:conf/fm/GorostiagaS21} and even extended to retroactive dynamic parametrization~\cite{DBLP:journals/corr/abs-2307-06763}.
Here, a parametrized stream has different configurations to access the complete history of the monitor.
Recently, refinement types for synchronous languages have been studied~\cite{DBLP:journals/corr/abs-2406-06221} that establish guarantees about the values of streams.
None of these approaches provide a type system reasoning about the timing of parametrized streams.
Yet, such a type system has proven to be highly beneficial\cite{DBLP:conf/cav/FaymonvilleFSSS19,DBLP:conf/cav/BaumeisterFSST20,DBLP:journals/tecs/BaumeisterFST19}.
%
%
%
%
%
%
%
%
\section{RTLola Syntax}
\label{sec:syntax}
In this section, we describe the syntax of \rtlola specifications in their abstract syntax tree (AST) format.
However, for readability, the example specifications in this paper are written in a more readable syntax that can be parsed to the syntax presented in the paper.

An RTLola specification consists of multiple stream definitions.
We differentiate between two types of streams: 
Input streams refer to system status readings, and output streams perform computations based on their stream declaration.
\begin{definition}[Specification]
\label{def:syntax}
    An \rtlola specification is a set of input streams and a set of output streams.
    Input streams and output streams have a unique identifier used in stream expressions.
    Output streams define their range of parameters and three stream declarations: one defining the start of a stream instance, one describing the stream computation, and one describing the stream's closing behavior.
    The grammar for a specification is given below:\\
    \scalebox{0.92}{\parbox{\linewidth}{
	\bgroup
	\def\arraystretch{0.5}%
	\begin{bnfgrammar}
		spec : Specifications ::= ($\texttt{input}^*$, $\texttt{output}^*$)
		;;
		input : Input streams ::= \sInput{id\_in}
		;;
		output : Output streams ::= \sOutput{id\_out}{\#para}{spawn}{eval}{close}
		;;
		\#para : Number of parameters \in \mathbb{N}
		;;
		spawn : Spawn declarations ::= \sSpawn{pac}{expr}{expr}
		;;
		eval : Eval declarations ::= \sEval{pac}{expr}{expr}
		;;
		close : Close declarations ::= \sClose{pac}{expr}
		;;
		pac : Pacing types ::= \syn{Any} $\mid$ \eventAc{ac} $\mid$ \globalFreq{freq} $\mid$ \localFreq{freq}
		;;
		freq : Frequencies \in \mathbb{N}
		;;
		ac : Activation conditions ::= id\_in $\mid$ ac \syn{and} ac $\mid$ ac \syn{or} ac
	\end{bnfgrammar}
	\egroup
	}}
\end{definition}
Parameterized output streams extend regular output streams with parameters, similar to functions. 
Each stream defines multiple instances, where each instance corresponds to a unique set of parameter values.
Output streams are specified through three declarations that define an instance's lifecycle and computation.
Each includes a pacing type and a semantic condition refining the timing.
The \emph{spawn} declaration specifies when a new instance is created and sets its parameters, the \emph{eval} declaration defines when and how the instance computes its value, and the \emph{close} declaration determines when the instance is removed.
Output streams without parameters are treated as parameterized streams with default values for their spawn, eval, and close clauses.

We refer to the first expression of each declaration as the \emph{When}-clause and the second, if present, as the \emph{With}-clause.
For the rest of the paper, we assume these expressions and the pacing of a declaration can be accessed via functions. For example, $\mathit{SpawnWith}(\sid)$ returns the With-clause of the spawn declaration of the stream identified by $\sid$.

A \textbf{pacing type} can either be a fixed global frequency, a fixed local frequency, or an event where an event is a combination of input streams receiving a new value simultaneously.
The $\pacingAny$ type is assumed if no pacing type is given.
\begin{definition}[Pacing Types]
    A pacing type defines the time points when an expression is evaluated.
	Their syntax is defined in \Cref{def:syntax}.    
\label{def:pacingTy}
\end{definition}
%
\textbf{Stream expressions} algorithmically define how a value is computed from other stream values, constants, and functions.
Their syntax is defined below, and their full definition can be found in \Cref{rtlola:syn:expr:full}.
They feature different access operators to refer to stream values. 
Each access includes a stream identifier and an expression to compute the target instance. 
A synchronous access retrieves the current value, a hold access returns the last value, and an offset access refers to a stream's history at discrete points. 
Aggregations provide sliding windows over past values, defined by a function and a real-time duration. 
The default expression handles failed accesses, like offsets at the monitor's start. 
Arithmetic expressions, constants, and functions are handled by the function expression, while a parameter access refers to the current instance's parameter value.
\scalebox{0.92}{\parbox{\linewidth}{
\bgroup
\def\arraystretch{0.5}%
\begin{bnfgrammar}
	expr : Expressions ::= \syncAccess{id}{p} $\mid$ \holdAccess{id}{p}
	| \offsetAccess{id}{p}{off} $\mid$ \default{expr}{expr}
	| \aggr{id}{p}{dur}{f\_a}
	| \functionExp{f}{$\texttt{expr}^*$} $\mid$ \constant{const}
	| \parameterAccess{p\_idx} $\mid$ \tuple{expr*}
	;;
	const : Constants \in \mathbb{V}
	;;
	id : Stream identifiers ::= id\_in $\mid$ id\_out
	;;
	id\_in : Input Stream identifiers \in \idsI
	;;
	id\_out : Output Stream identifiers \in \idsO
	;;
	p : Parameters \in \texttt{expr} \times \dots \times \texttt{expr}
\end{bnfgrammar}
\egroup}}
\section{Relational Semantics}
\label{sec:semantics}
The semantic of an \rtlola specification is defined as a relation between the value in an evaluation model and the value computed by the stream declarations.
An \emph{Evaluation Model} $\world \in \World$ is the combination of a $\StreamMap$ and a $\stimeMap$.
\[
\begin{array}{r l}
	\Stream &:= \instanceId \rightarrow \stime \rightarrow \svalue\\
	\StreamMap &:= \sr \rightarrow \Stream\\
	\stimeMap & := \stime \rightarrow \mathbb{R}\\
	\World & := \StreamMap \times \stimeMap
\end{array}
\]
The $\StreamMap$ assigns each stream and instance identifier to an infinite timed sequence of optional values.
The $\stimeMap$ is a total function from a discrete timestamp to a real-time value later used to reason over real-time accesses.
%
We overload the access function to evaluation models for convenience based on its input.
Give an evaluation model $\world \in \World$, a stream access to a stream $\sid \in \sr$ and a  discrete timestamp $t \in \stime$ it holds:
	\[
	\begin{array}{r l l}
		\world(\sid) &:= \mathit{stream\_map}(\sid) &\mbox{if } \world = (\mathit{stream\_map}, \mathit{time\_map})\\ 
		\world(t) &:= \mathit{time\_map}(t) &\mbox{if } \world = (\mathit{stream\_map}, \mathit{time\_map})
	\end{array}
	\]
To validate an evaluation model, the semantics assert that the time is monotone and that each value in the model matches the computed value described by the stream declarations.
\begin{definition}[Specification Semantics]
	Let $\varphi$ be a \rtlola specification. The set of valid evaluation models for the specification is
	\[
		\llbracket \varphi \rrbracket := \left\{ \world \mid \forall \sid \in \srout. \world \vDash \varphi(\sid) \wedge \forall t \in \stime. \world(t) < \world(t + 1)\right\}.
	\]
	An evaluation model $\world$ satisfies an output stream iff
	\[
	\begin{array}{c}
		\world \vDash \sOutput{\sid}{p}{s}{e}{c} := \\ \forall t \in \stime, \instid \in \mathbb{V}^p. \evalTo{\sOutput{\sid}{p}{s}{e}{c}}{\world(\sid)(\instid)(t)}
	\end{array}
	\]
\end{definition}
For the evaluation of the stream declarations, a set of inference rules define $\evalTo{}{}$ computing their value given a model $\world$, a discrete timestamp $t$, and the current stream instance $\instid$.
These inference rules evaluate the pacing and use the prefix of a stream instance, defined in the following subsection.

\subsubsection{Pacing Evaluation.}
\label{sec:semantics:pacing}
The evaluation of a stream pacing distinguishes four alternatives:
\pacingAny~is a condition that is always true, the \emph{Activation Condition} bounds the pacing to the updates of input data, and the two frequency types (global and local) describe fixed points in real-time.
The difference between the two types of frequencies is their starting point.
A \emph{global frequencies} starts with the beginning of the monitor, whereas the start of a \emph{local frequencies} is bounded by the condition described with the \emph{Spawn} definition of a stream.

	Given a evaluation model $\world \in \World$, a discrete timestamp $t \in \stime$, and discrete starting time $\tstart \in \stime$, the evaluation $\pacing{t}{}{\tstart}$ of a pacing is
	\[
	\begin{array}{r l}
		\pacing{t}{}{\tstart} &: \mathit{Pacing~types}\rightarrow \mathbb{B}\\
		\pacing{t}{\syn{\pacingAny}}{\tstart} &:= \top\\
		\pacing{t}{\eventAc{\mathit{ac}}}{\tstart} &:= \pacingAC{t}{\mathit{ac}}\\
		\pacing{t}{\globalFreq{\mathit{freq}}}{\tstart} &:= \worldaccone{t} \mathbin{\texttt{mod}} \mathit{freq} = 0\\
		\pacing{t}{\localFreq{\mathit{freq}}}{\tstart} &:= (\worldaccone{t} - \worldaccone{\tstart}) \mathbin{\texttt{mod}} \mathit{freq} = 0,\\
	\end{array}
	\]
	where the evaluation of an \emph{activation condition} $\mathit{ac}$ is
	\[
	\begin{array}{r l}
		\pacingAC{t}{} &: \mathit{Activation~conditions}\rightarrow \mathbb{B}\\
		\pacingAC{t}{sr} & := \worldacctwo{sr}{t} \neq \bot\\
		\pacingAC{t}{\lhs~\syn{and}~\rhs} & := \pacingAC{t}{\lhs} \wedge \pacingAC{t}{\rhs}\\
		\pacingAC{t}{\lhs~\syn{or}~\rhs} & := \pacingAC{t}{\lhs} \vee \pacingAC{t}{\rhs}.
	\end{array}
	\]

\subsubsection{Prefix}
\label{sec:semantics:pacing}
The prefix of a stream instance is a finite sequence of stream values $\mathbb{V}$ up to a discrete timestamp $t$.
It allows expression evaluations to access past values or aggregate over the instance's full or partial history.
However, the prefix does not always start with the instance's first value, as the close declaration causes the monitor to remove past values from its memory.
The beginning is defined as the first time the instance is spawned and not closed until $t$.
We define $\isSpawned$ and $\isClosed$ predicates to determine when an instance starts and stops, with the full definition in \Cref{def:is_spawned}. 
%
%

Based on this, the \textit{AliveSince} predicate determines if a stream instance currently exists.
Given an evaluation model $\world \in \World$, a stream identification $\sid \in \sr$, an instance identification $\instid \in \instanceId$, and a timestamp $t \in \stime$, the first timestamp such that the stream instance is spawned and not closed until $t$ is described by:
	\[
	\begin{array}{l}
		\alive : \World \times \sr \times \instanceId \times \stime \rightarrow \stime_{\bot} \\
		\alive(\world, \sid, \instid, \tend)\\
		~~:= \left\{
			\begin{array}{lll}
				0 && \mbox{if~} \sid \in \idsI\\
				\tstart & \mbox{if } \tstart = \mathit{min}\{t \mid &\isSpawned(\world, \sid, \instid, t)\\
				&&\wedge~\forall t \leq t' < \tend. \lnot \isClosed(\world, \sid, \instid, t')\}\\
				\bot & \mbox{otherwise}
			\end{array}
			\right.
	\end{array}
	\]
A \textit{Prefix} then collects all non-optional values up to the start of the prefix or returns an empty sequence if the instance is not alive at time $t$.
\[
\begin{array}{l}
	\prefix : \World \times \sr \times \instanceId \times \stime \rightarrow (\stime, \mathbb{V})^* \\
	\prefix(\world , \sid, \instid, \tend)\\~:= \left\{\begin{array}{ll}
		\prefix'(\world(\sid)(\instid), \tstart, \tend) & \mbox{if~} \alive(\world, \sid, \instid, t) = \tstart 
		\\ & \wedge~\tstart \neq \bot\\
		\epsilon & \mbox{otherwise}
  \end{array}\right.
\end{array}
\]
where $\prefix'$ returns the sequence of non-optional values, from $\tstart$ to $\tend$:
	\[
\begin{array}{l}
	\prefix' : (\stime \rightarrow \svalue) \times \stime \times \stime \rightarrow (\stime, \mathbb{V})^*\\
	\prefix'(\instance, \tstart, \tend)\\~:= \left\{\begin{array}{ll}
    (\tend, v) & \mbox{if $\tend = \tstart \wedge \instanceacc{\tend} = v$}\\
    &~~~\wedge~v \in \mathbb{V}\\
    \epsilon & \mbox{if $\tend = \tstart \wedge \instanceacc{\tend} = \bot$}\\
    \prefix'(\instance, \tstart, \tend - 1) \cdot (\tend, v) & \mbox{if $\instanceacc{\tend} = v \wedge v \in \mathbb{V}$}\\
    \prefix'(\instance, \tstart, \tend - 1) &\mbox{if $\instanceacc{\tend} = \bot$}
    \end{array}\right.
\end{array}
\]
For \textbf{window aggregations} of a real-time duration, we then restrict the prefix of a stream instance to a window slice, starting from a real-time timestamp.
Given a stream instance prefix $\mathit{prefix}$, and a real-time timestamp $\tstart \in \mathbb{R}$, the window of a stream instance is defined as:
\[
	\begin{array}{l}
		\window: \World \times (\stime, \mathbb{V})^* \rightarrow \mathbb{R} \rightarrow (\stime, \mathbb{V})^*\\
		\window(\world, \mathit{prefix}, \tstart)\\:=\left\{
			\begin{array}{ll}
			\window(\world, \mathit{prefix}', \tstart) \cdot (t,v) & \mbox{if~} \mathit{prefix} = \mathit{prefix}' \cdot (t,v) \wedge \world(t) > \tstart\\
			\epsilon & \mbox{otherwise}
			\end{array}\right.
	\end{array}
\]

\subsubsection{Expression Evaluation.}
We provide a set of inference rules defining the evaluation of stream expressions.
Given an evaluation model $\world$, a discrete timestamp $t$, and an instance context, i.e., under which instance the current stream expression is evaluated, the evaluation $\evalTo{}{}$ of a stream expression is defined by the set of inference rules defined in \Cref{fig:sema:expressions}.

Consider the three rules in \Cref{fig:sema:expression} as an example defining the evaluation of the \syn{Sync} and \syn{Hold} expression.
Both expressions refer to values of stream instances described by the $\sid$ and $\instid$.
\begin{figure}[t]
	\centering
	\scalebox{0.9}{
	\begin{mathpar}
	\inferrule[eval-expr-syn]{\evalTo{\mathit{expr}}{\mathit{inst'}}\\\worldaccthree{\sid}{\mathit{inst'}}{t} = v}{\evalTo{\syncAccess{\sid}{\mathit{expr}}}{v}}\qquad\\
\inferrule[eval-expr-hold]{\evalTo{\mathit{expr}}{\mathit{inst'}} \\ \mathit{prefix} = \prefix(\world,\sid, \mathit{inst'}, t) \\ \mathit{prefix} = \mathit{prefix}' \cdot (t',v)}{\evalTo{\holdAccess{\sid}{\mathit{expr}}}{v}}\qquad\\
        \inferrule[eval-expr-hold-dft]{\evalTo{\mathit{expr}}{\mathit{inst'}} \\ \mathit{prefix} = \prefix(\world,\sid, \mathit{inst'}, t)\\\mathit{prefix} = \epsilon}{\evalTo{\holdAccess{\sid}{\mathit{expr}}}{\bot}}\qquad\\	
\end{mathpar}}
\caption{Inference rule examples for stream expressions}
\label{fig:sema:expression}
\end{figure}
Intuitively, the synchronous expression accesses the current value in the evaluation model of a stream instance.
The corresponding inference rule first determines the concrete stream instance with the evaluation of the second argument before accessing the value of this instance in the evaluation model.
The hold access has two rules that both compute the prefix of the stream instance to determine if a last value exists.
If this prefix is not empty, the first inference rule evaluates the hold access to this last value.
Otherwise, the second rule that evaluates the expression to $\bot$ is applicable.
\subsubsection{Stream Instance Evaluation.}
Using the expression and pacing evaluation, a set of inference rules is given, evaluating an output stream instance.
This set of inference rules covers the different combinations of the evaluation pacing, the dynamic filters $\mathit{EvalWhen}$, and stream instances that might not be alive.
Generally, a stream instance evaluates to a value $v \in \mathbb{V}$, iff the stream instance is alive, and the evaluation pacing and the dynamic filter evaluate to true.
\begin{figure}[t]
\centering
	\scalebox{0.9}{
	\begin{mathpar}
    	\inferrule[eval-instance-value]{\alive(\world, \sid, \instid, t) = \tstart\\ \tstart \neq \bot \\\pacing{t}{\pac}{\tstart} \\ v \in \mathbb{V} \\ \evalTo{\mathit{when}}{true}\\\evalTo{\mathit{with}}{v}}{\evalTo{\sOutput{\sid}{\_}{\_}{\sEval{\pac}{\mathit{when}}{\mathit{with}}}{\_}}{v}}\qquad\\
    	\inferrule[eval-instance-false-eval-when]{\alive(\world, \sid, \instid, t) = \tstart\\ \tstart \neq \bot \\\pacing{t}{\pac}{\tstart} \\ v \in \mathbb{V} \\ \evalTo{\mathit{when}}{false}}{\evalTo{\sOutput{\sid}{\_}{\_}{\sEval{\pac}{\mathit{when}}{\mathit{with}}}{\_}}{\bot}}\qquad\\
	\end{mathpar}}
	\caption{Inference rules for the stream instance evaluation}
	\label{fig:sema:instances}
\end{figure}
Given an evaluation model $\world$, a discrete timestamp $t$, and an instance context, the evaluation $\evalTo{}{\mathit{inst}}$ of a stream instance is defined by the inference rules in \Cref{fig:sema:instances:complete} partially shown in \Cref{fig:sema:instances}.
\section{The RTLola Typesystem}
\label{sec:types}
The semantics described in \Cref{sec:semantics} handles stream accesses fallibly, meaning they can result in undefined values.
For some accesses, like $\holdAccess{\sid}{\expr}$ or $\offsetAccess{\sid}{\expr}{o}$, this is expected, and the specifier must provide a default value.
However, a synchronous access explicitly requires that the accessed value exists, and thus, it should never result in an undefined value.
In the next section, we present \rtlola's refinement type system, which ensures this condition is met.

Consider the following example:
\begin{lstlisting}
input a
output b @1Hz@ := 42
output c @a@ := b
\end{lstlisting}
In this example, the stream \lstinline{c} is evaluated whenever \lstinline{a} gets a new value.
Yet, it is not guaranteed that the stream \lstinline{b} is evaluated simultaneously so that its value can be accessed by \lstinline{c}.
This is only the case when the input stream \lstinline{a} receives values in one-second intervals.
In practice, however, inputs are controlled by the system, and the monitor should not assume anything about their value or timing.

Similar to programming languages, \rtlola features a value-type analysis, ensuring functions are called with appropriate values.
This paper assumes all specifications are well-typed based on this analysis.
For completeness, the definition of the value type system is included in \Cref{appendix:value_types}.
For an extended version of all type analyses, we refer to the dissertation of Maximilian Schwenger \cite{Schwenger_2022}.

This section proceeds as follows: \Cref{def:spec_safety} formalizes the condition that the semantics cannot fail.
From that, we build lemmas and theorems that relate well-typedness to semantic safety, i.e., every well-typed specification cannot fail at runtime.
Finally, we present the inference rules that define the type system.



\subsection{Type checking soundness}
We present a series of theorems establishing the soundness of the type-checking procedure.
The complete proofs of the theorems and lemmas are presented in \cref{app:proofs}.
First, we define specification safety.
\begin{definition}[Specification Safety]
\label{def:spec_safety}
	A specification is safe iff it's semantics cannot fail, i.e., there is at least one world that is consistent with the specification $\varphi$ for any sequence of inputs:
	\begin{center}
		\scalebox{0.94}{$
		\textit{safe}(\varphi) \Longleftrightarrow (\forall i \in \idsI \rightarrow \stime \rightarrow \svalue.\ \exists \world \in \World.\ \forall s \in \idsI.\ \world(s)(\epsilon) = i(s) \land \world \in \semEval{\varphi})
		$}
	\end{center}

\end{definition}
Next, we define an $\mathit{active}$ predicate that semantically defines the points in time when a stream instance evaluates to a defined value.
\begin{definition}[Stream-Activation]
	\label{def:stream-activation}
	A stream instance is declared active if it is alive, its semantic filter evaluates to true, and its pacing type holds:
	\begin{align*}
		\mathit{active}(\sid, t, \instance, \world) \Longleftrightarrow &\alive(\world, \sid, \instance, t) = \tstart\\ 
		&\land \tstart \neq \bot \land \evalTo{\evalWhen{\sid}}{true}\\
		&\land \pacing{t}{\evalPac{\sid}}{\tstart}
	\end{align*}
\end{definition}
The next lemma explores the criterion under which the semantic can fail.
\begin{lemma}[Failable Semantics]
\label{lemma:failable_semantics}
The semantics for a specification fails iff there is a synchronous access that produces $\bot$.
\end{lemma}
The proof of this lemma is a structural induction over the syntax of \rtlola.
We determine that the \emph{Eval-Expr-Syn} inference rule is the only one that can produce $\bot$ during the induction.
Based on this lemma, abstract safety captures the criterion under which a synchronous accesses cannot fail.
\begin{theorem}[Abstract-Safety]
	\label{theorem:abstract-safety}
	Given a specification $\varphi$.
	Let $\syncAccess{\sid}{expr}$ be a synchronous access in the stream $\sOutput{\sid'}{n_p}{s}{e}{c}$.
	\begin{align*}
		(\forall t \in \stime, \world \in \World, \instance' \in \svalue^{n_p}.&\\
		\mathit{active}(\sid',t, \instance', \world) &\land \evalToAt{t}{\instance'}{expr}{\instance} \implies \mathit{active}(\sid, t, \instance, \world))\\
		&\implies \textit{safe}(\varphi)
	\end{align*}
\end{theorem}
This theorem follows by contradiction of the definition of the $\mathit{active}$ predicate with the \emph{Eval-Expr-Syn} inference rule using \Cref{lemma:failable_semantics}.
In the final step, we derive type-checked safety as a sufficient criterion for abstract safety.
\begin{theorem}[Typechecked-Safety]
\label{theorem:typechecked-safety}
	Let the predicate $\mathit{types}(\varphi)$ represent that every stream in the specification $\varphi$ typed correctly according to the inference rules present below.
	Then for any synchronous access $\syncAccess{\sid}{expr}$ within the stream $\sOutput{\sid'}{n_p}{s}{e}{c}$ of $\varphi$ it holds that:
	\begin{align*}
		\mathit{types}(\varphi) &\implies \forall t \in \stime, \world \in \World, \instance' \in \svalue^{n_p}.\\
		&\mathit{active}(\sid',t, \instance', \world) \land \evalToAt{t}{\instance'}{expr}{\instance}\implies \mathit{active}(\sid, t, \instance, \world)
	\end{align*}
\end{theorem}
The theorem follows directly from the type inference rules for the synchronous access and the output stream.

%
%
\subsubsection*{The Pacing Type System.}
The \rtlola pacing type system is a refinement type system that reasons about the time points at which a stream instance evaluates.
Instead of representing traditional data types like integers or strings, a pacing type defines a set of discrete real-time time points. 
For example, the pacing type $@\texttt{Global}(1Hz)$ represents whole seconds.
Semantic types are boolean stream expressions that refine this set based on their evaluation:
a time point stays in the set, iff the expression evaluates to \texttt{true} at that time.
Thus, the pacing type $@\texttt{Global}(1Hz)\texttt{ when }x > 5$ represents the set of whole seconds when the stream x evaluates to a value greater than $5$.
To simplify the reasoning, we split the pacing type and semantic type into separate lattices and inference rules.
These lattices are bounded meet-semilattices~\cite{SWB-21309293X} and are defined for an arbitrary but fixed specification.
In this subsection, these lattices are defined through their $\sqsubseteq$ "more concrete" relation on the set of types.

The pacing-type lattice has two components: event-based and periodic types.
Event-based types represent events over input streams as positive boolean formulas, with propositions evaluating to true when an input stream updates; their order is defined by boolean implication.
Periodic types represent fixed frequencies, ordered by the greatest common divisor of their periods, and are classified into local and global frequencies.
Global frequencies begin at the monitor's start, while local frequencies start when the instance is created.
Consider the example:
\begin{lstlisting}[label={ex:global_local_f}]
	input a: Int
	output b @1Hz@ := 5
	output c spawn when a > 42 eval @1Hz@ with b
\end{lstlisting}
In this specification, the output stream \lstinline{b} produces a constant value in periodic intervals aligned to the start of the monitor.
The output \lstinline{c} also produces values at this rate, but this it starts after the condition \lstinline{a > 42} evaluates to $true$.
As a result, the stream \lstinline{b} has the pacing type $\gf{1}$ while stream \lstinline{c} has $\lf{1}$.
These frequencies of \lstinline{b} and \lstinline{c} can be out of sync, and the access from \lstinline{c} to \lstinline{b} should be rejected by the type-checker.
%
While the pacing type lattice allows for coercing all local frequencies, the semantic type analysis performs an analysis of these local periodic timings based on their spawn declaration.
\\
\textbf{The Pacing Type Lattice.}
    The set of pacing types is defined as:
    \[
        \PT := \{\periodic, \gf{x}, \lf{x}, \eventT{e}, \top, \bot \mid x \in \mathbb{N}, e \in \ac \}    
    \]    
    \[
        \ac \ni \phi, \psi := \idsI \mid \phi \land \psi \mid \phi \lor \psi
    \]
    The relation $\ptCon$ is minimal with regard to the following equations:
    \begin{alignat*}{2}
        \pt &\ptCon \top &&\Longleftrightarrow \pt \in \PT\\
        \bot &\ptCon \pt &&\Longleftrightarrow \pt \in \PT\\
        \gf{x} &\ptCon \periodic &&\Longleftrightarrow x \in \mathbb{N}\\
        \lf{x} &\ptCon \periodic &&\Longleftrightarrow x \in \mathbb{N}\\
        \gf{x_1} &\ptCon \gf{x_2} &&\Longleftrightarrow \exists k \in \mathbb{N}. k \cdot x_2 = x_1\\
        \lf{x_1} &\ptCon \lf{x_2} &&\Longleftrightarrow \exists k \in \mathbb{N}. k \cdot x_2 = x_1\\
        \eventT{e_1} &\ptCon \eventT{e_2} &&\Longleftrightarrow e_2 \implies e_1\\
    \end{alignat*}
Note that while the set $\PT$ is generally infinite, it is finite for all types appearing in a single specification.
We associate a separate pacing type to each stream declaration, resulting in a triple of inferred types.
We refer to the spawn pacing using $\pt_s$, to the eval pacing with $\pt_e$, and to the close pacing with $\pt_c$.
We use $\ptt$ to refer to the triple of the above types for which the $\ptCon$ relation is defined as:
\[
	\ptt \ptCon \ptt' \Longleftrightarrow \pt_s \ptCon \pt_s' \land \pt_e \ptCon \pt_e' \land \pt_c' \ptCon \pt_c 
\]
All pacing type inference rules are depicted in \Cref{appendix:pacing_types}.
To give an intuition for these rules, we discuss the \emph{pacing-output} inference rule:
	\[
		\scalebox{0.9}{
		\inferrule[pacing-output]
		{
			(\top, \pt_e^s, \top) \types{} \sSpawn{\pt_s'}{e_s^f}{e_s}\\
			(\pt_s^e, \pt_e^e, \pt_c^e) \types{} \sEval{\pt_e'}{e_e^f}{e_e}\\
			(\pt_s^c, \pt_e^c, \pt_c^c) \types{} \sClose{\pt_c'}{e_c^f}\\
			\pt_s^c \in \{\top, \pt_e^s\}\\
			\pt_c^c \in \{\top, \pt_e^c\}\\
			\pt_e^e = \lf{x} \rightarrow \pt_e^s = \pt_s^e \land \pt_e^c = \pt_c^e \\
			\pt_e^c = \lf{x} \rightarrow \pt_e^s = \pt_s^c \land \pt_e^c = \pt_c^c\\
			\pt_e^s \ptCon \pt_s^e\\
			\pt_c^e \ptCon \pt_e^c\\
		}
		{(\pt_e^s, \pt_e^e, \pt_e^c) \types{} \sOutput{\sid}{p\_def}{\sSpawn{\pt_s'}{e_s^f}{e_s}}{\sEval{\pt_e'}{e_e^f}{e_e}}{\sClose{\pt_c'}{e_c^f}}}
	}\]
It is applied to output streams and first infers pacing types for the \emph{Spawn}, \emph{Eval}, and \emph{Close} declaration.
For each declaration, a spawn, eval, and close pacing is inferred, indicated by the subscript of the tuple elements.
These types are inferred from the \emph{when} and \emph{with} expressions of the declaration, while the eval pacing of a clause is further restricted to its annotated type.
The resulting pacing types are the eval pacings of the respective declaration.

For the spawn declaration, we demand that its spawn and close pacing is $\top$.
For the close declaration, the same conditions apply with one exception:
It can run at the same local frequency as the stream.
Hence, its spawn and close pacing is only permitted to be either $\top$ or equal to the pacing of the stream.
Lastly, the inference rule demands that the spawn and close pacing of the streams synchronously accessed in the eval declaration are compatible with the spawn and close pacings inferred for the output stream.
This condition is strengthened from compatibility to equality if the pacing of the eval declaration is a local frequency to prevent shifted periodicities, as shown in the example above.

Shifted periodicities also cause the example in \Cref{sec:intro:example} to be not well typed.
There, the shift is due to the added \textit{spawn when} annotation of the trigger.
\subsubsection*{The Semantic Type Lattice.}
Semantic types allow the refinement of pacing types depending on runtime information.
Defined as boolean stream expressions, their type lattice is given through the semantic implication of stream expressions.
The set of semantic types is defined as $\ST := \{ Expr(e), \top, \bot \mid e \in \exprs \}$ and
the relation $\stCon$ is minimal with regard to the following equations:
    \begin{alignat*}{2}
        \st &\stCon \top &&\Longleftrightarrow \st \in \ST\\
        \bot &\stCon \st &&\Longleftrightarrow \st \in \ST\\
        Expr(e_1) &\stCon Expr(e_2) &&\Longleftrightarrow \forall t \in \stime.\semEval{e_2}^t \rightarrow \semEval{e_1}^t
    \end{alignat*}
Where $\semEval{e}$ defines the evaluation of $e$ in its appropriate context according to the semantic presented in \cref{sec:semantics}.

As with pacing types, we simultaneously reason about the semantic type of the spawn, eval, and close declaration.
We use $\stt$ to refer to the triple of semantic types $\st_s$, $\st_e$ and $\st_c$ for which the $\stCon$ relation is defined as follows:
\[
	\stt \stCon \stt' \Longleftrightarrow \st_s \stCon \st_s' \land \st_e \stCon \st_e' \land \st_c' \stCon \st_c
\]
The inference rules for semantic type checking are depicted in \Cref{appendix:semantic_types}.
If necessary, the id of the stream containing the expression is annotated as a context to the $\types{}$ relation.
Below, we provide an intuition for semantic type inference using the \emph{Sem-Sync-Output} rule.
	\[
	\scalebox{0.9}{
		 \inferrule[Sem-Sync-Out]
    	{
    		\sOutput{\sid}{n_p}{s}{e}{c} \in \outputs\\
    	    \sOutput{\sid'}{n_p'}{s'}{e'}{c'} \in \outputs\\
            \forall 1 \leq k \leq n.\ e_k = \parameterAccess{i} \land k \paramEq{\sid}{\sid'} i\\
            \stt \types{} \sOutput{\sid}{n_p}{s}{e}{c}\\
            \stt'\stCon \stt
        }
    	{\stt' \types{\sid'} \syncAccess{\sid}{(e_1,...,e_n)}}
    }
	\]
This inference rule types a synchronous access from the output stream with id $\sid'$ to the output stream with id $\sid$.
It requires that the resulting type(s) of the access $\stt'$ is more concrete than the inferred type(s) of the accessed output stream.
The universally quantified statement poses an additional requirement for synchronous accesses, which we call parameter equality.
%

Ensuring that the timing of the accessing and accessed instance is compatible is, in fact, not sufficient.
Consider the following example:
\\
	\begin{minipage}{0.5\textwidth}
	\begin{lstlisting}
    input i1: Int64
    input i2: Int64
    output a(p1: Int64) spawn with i1 ...
    output b(p2: Int64, p3: Int64, p4: Int64)
        spawn with (i2, i1, i1)
        eval with
            a(p2) // Invalid
            a(p3) // Valid
            a(p4) // Valid
	\end{lstlisting}
	\end{minipage}%
	\begin{minipage}{0.45\textwidth}
	While the spawn and evaluation pacing type of stream \streamName{b} is compatible with the pacing of stream \streamName{a}, the access in line 7 is still invalid.
	The inputs \streamName{i1} and \streamName{i2} are not guaranteed to receive the same value from the system.
	\end{minipage}
	\\
Yet, the accesses in lines 8 and 9 are valid because the parameters \streamName{p3} and \streamName{p4} are both instantiated by \streamName{i1} in the streams spawn expression, which is also true for the parameter \streamName{p1} of the accessed stream, which guarantees that they always evaluate to the same value.

Intuitively, the pacing and semantic types guarantee that there exists any instance of accessed stream, while it does not necessarily ensure that the \textit{accessed} instance, identified by the expressions constituting the parameter values after the accessed stream name, exists.
For this, two conditions must be met:
First, all expressions that constitute the parameter values of the accessed stream are parameter accesses to parameters of the accessing stream.
Second, these \textbf{parameters are equal}, meaning they are instantiated by the same stream expression:
\\
\noindent Let $o^1 = \sOutput{\sid^1}{n}{\sSpawn{pt^1}{e_f^1}{(e_1^1, ..., e_n^1)}}{e^1}{c^1}$ be an output stream as well as $o^2 = \sOutput{\sid^2}{m}{\sSpawn{pt^2}{e_f^2}{(e_1^2, ..., e_m^2)}}{e^2}{c^2}$.
    Then:
    \[
         \forall 1 \leq i \leq n. \forall 1 \leq j \leq m.\ i \paramEq{\sid^1}{\sid^2} j \Longleftrightarrow e_i^1 =_s e_j^2
    \]
    Where $=_s$ refers to the syntactic equality of expressions.

For streams with compatible pacing, parameter equality statically guarantees that parameters of different streams have the same value at runtime.
By restricting the expressions used to identify the accessed instance in a synchronous access to parameters of the accessing stream that match the parameters of the accessed stream, the \emph{Sem-Sync-Output} inference rule guarantees that the exact instance being accessed exists.

\subsection{Limitations of the Type-Checking procedure.}
In this section, we elaborate on the limitations of the type inference.
First, we analyze the decidability of stream expression implication as used in the semantic type lattice and in \rtlola in general.
We presented an undecidability result and show how it can be bypassed by over-approximating the semantic type lattice. 
\begin{theorem}[\rtlola Undecidability]
	\label{theorem:undecidability}
	The satisfiability of the syntactic fragment of well-formed \rtlola specifications is undecidable.
\end{theorem}
\begin{proof}
We derive the undecidability from the essential undecidability of Robinson arithmetic~\cite{robinson1950essentially}.
Let $\varphi$ be an arbitrary well-formed specification.
Since the specification is well-formed, only one evaluation model exists for each input sequence, resulting in the consistency criteria.
From the essentially undecidability of the Robinson Arithmetic, it follows that the satisfiability of stream expressions is undecidable, as they are a consistent extension of the Robinson Arithmetic.
Consequently, the satisfiability of \rtlola as a whole is undecidable.
\end{proof}
\Cref{theorem:undecidability} derives that the implication of stream expression is undecidable. 
In the implementation of the type-checker~\cite{frontend}, the stream-expression implication is therefore over-approximated through syntactic equivalence of the expressions.
One example where this over-approximation leads to the rejection of specifications is when an arithmetic theory is required to reason about an implication, e.g., the type inference cannot assure that $i > 7 \rightarrow i > 5$.
\section{Well-defined Parameterized Specifications}
\label{sec:welldefined}
The type system in \cref{sec:types} ensures the existence of an evaluation model for any input sequence.
Yet, the semantics still allow for multiple evaluation models.
This could lead to non-determinism of the monitor, undesirable in a safety-critical environment.
We follow the definitions in Lola~\cite{d2005lola} to identify specifications with deterministic behavior.
\begin{definition}[Well-definedness]
	A \rtlola specification is \emph{well-defined} iff for any set of input streams it has exactly one evaluation model.	
\end{definition}
Since this definition cannot be checked syntactically, we define a syntactic sufficient criterion analyzing the dependencies between streams.
%
%
%
%
Given an \rtlola specification $\varphi$.
A \emph{dependency graph} for $\varphi$ is a weighted and directed multi-graph $G = (V, E)$, with
	\[
	\begin{array}{r l}
		V &= \sr\\
		E &\subseteq V \times (L \times A) \times V\\
		L &= \{\mathit{Spawn}, \mathit{EvalWhen}, \mathit{EvalWith}, \mathit{Close} \}\\
		A &= \{\mathit{Sync} \} \cup \{\mathit{Hold}\} \cup (\mathit{Offset} \times \mathbb{N}\}) \cup \{(\mathit{Aggr}, \mathbb{R})\}
	\end{array}
	\]
	Where $E$ satisfies: For any expression $\mathit{expr}$ in a stream declaration $\sid \in \srout$, given its location $l \in L$ it holds:  
	\[
	\begin{array}{r l}
		\mathit{expr} = \syncAccess{\sid'}{p} &\implies (\sid, (l,\mathit{Sync}), \sid') \in E\\
		\mathit{expr} = \offsetAccess{\sid'}{p}{\mathit{off}} &\implies (\sid, (l,(\mathit{Offset}, \mathit{off})), \sid') \in E\\
		\mathit{expr} = \holdAccess{\sid'}{p	} &\implies (\sid, (l,\mathit{Hold}), \sid') \in E\\
		\mathit{expr} = \aggr{\sid'}{p}{\mathit{dur}}{f} &\implies (\sid, (l,(\mathit{Aggr}, \mathit{dur})), \sid') \in E\\
	\end{array}
	\]
Compared to stream equations in Lola, parameterized streams contain multiple stream expressions.
Hence, we include the location of the accesses in the edges to encode the constraints between these locations, e.g., the close expression given at time $t$ does not affect the evaluation of a stream at $t$.
Additionally, the different accesses need to be represented in the dependency graph.
We define the well-formedness criteria of a specification based on its dependency graph.
\begin{definition}[Wellformedness]
	\label{def:wellformed}
	A specification $\varphi$ is \emph{well-formed} iff all cycles $\{(v_1, (l_1, a_1), v_2)\dots(v_n, (l_n, a_n), v_1)\} \in E^*$ in the dependency graph $G = (V, E)$ satisfy:
	\[
	(\exists i. l_i = \mathit{Close}) \vee (\forall i. l_i = \mathit{EvalWith} \wedge \exists i'. a_{i'} = (\mathit{Offset}, o) \wedge o \neq 0)
	\]
\end{definition}
With this adapted definition of well-formedness, we can prove the same theorem as in the original Lola paper~\cite{d2005lola}.
\begin{proposition}[Well-Formedness implies Well-Definedness]
	Every well-formed \rtlola specification is well-defined.
\end{proposition}
\section{Evaluation}
\label{sec:evaluation}
We evaluate the performance of the type analysis using realistic specifications from the aerospace domain \cite{volostream,DBLP:conf/fm/BaumeisterFKS24}.
Furthermore, we use artificially generated specifications that scale in the number of streams to quantify its scalability.
All specifications are included in \Cref{appendix:benchmarks} of this paper.

The experiments are run natively on a MacBook Pro from 2020 equipped with an M1 processor and 16Gb of RAM.
We use the criterion~\cite{criterion} micro benchmarking library to automatically prime and sample the median runtime of the type analysis from at least 100 executions.

For the realistic benchmarks, \Cref{fig:realistic_bench} demonstrates that the performance impact on real-world complex specifications like a geofence is negligible with the analysis taking less than one second.

All three artificial specifications are crafted to be arbitrarily scaled and represent a worst-case scenario for the analysis.
The first specification evaluates the performance of the type checking analysis based on an increasing number of streams that synchronously access each other in a chain.
The second one evaluates the performance impact of parameters by adapting each stream in the previous case with an increasing number of parameters.
Lastly, we test the performance of the semantic type checking by introducing \emph{when} statements to the \emph{eval} declarations, adding a conjunct for each step.

\begin{figure}[t]
	\begin{subfigure}{0.45\textwidth}
		\bgroup
		\footnotesize
		\def\arraystretch{1.5}%
		\setlength\tabcolsep{0.25em}
		\begin{tabular}{l  r  r}
			Specification & \# Lines & Time [ms]\\
			\hline
			Watchdog & 9 & 0.69\\
			RCC & 10 & 0.84 \\
			FFD & 15 & 2.82\\
			Intruder & 14 & 4.76\\
			Waypoints & 15 & 5.51\\
			Geofence & 86 & 611.32\\
		\end{tabular}
		\egroup
		\vspace{1em}
		\caption{Performance results of the aerospace benchmarks.}
		\label{fig:realistic_bench}
	\end{subfigure}
	\hfill
	\begin{subfigure}{0.5\textwidth}
		\includegraphics[width=\textwidth]{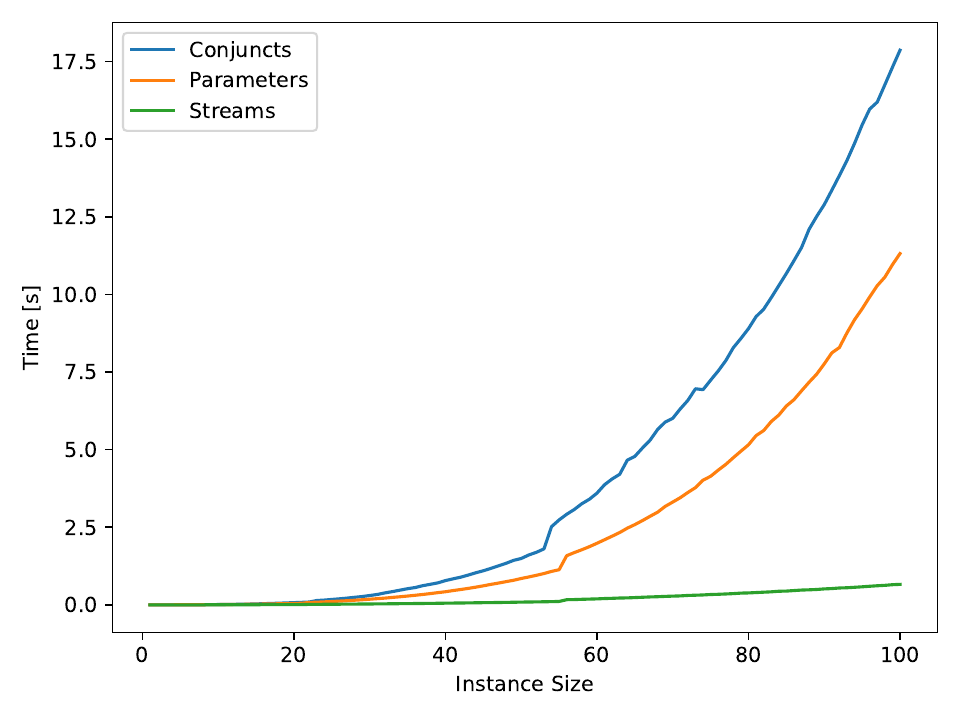}
		\caption{Performance results of the artificial benchmarks.}
		\label{fig:artifical_bench}
	\end{subfigure}
	\caption{Evaluation Results}
	\label{fig:eval_results}
\end{figure}

The artificial performance benchmarks in \Cref{fig:artifical_bench} indicate that adding parameters to streams or conjuncts to evaluation declarations impacts the performance of the type analysis compared to the \emph{Streams} benchmark without any of them.
While the runtime seems to increase exponentially with the size of the instance, it is most likely a quadratic relation stemming from the underlying constraint-solving algorithm in the type checker.
Overall, even with as many as 100 streams and equally many conjuncts in the largest expression, the runtime of the type analysis did not exceed 18 seconds.

\section{Conclusion}
\label{sec:conclusion}

This paper introduced parameterized streams to \rtlola, enabling the handling of infinite data domains in real-time environments.
We provided a formal foundation for using these streams to monitor safety-critical cyber-physical systems, such as unmanned aerial vehicles.
To ensure reliability, we identified a safe subset of \rtlola specifications that are guaranteed not to fail at runtime.
This was achieved through a novel refinement type system that statically verifies these specifications.
Unlike traditional type systems in programming languages, this approach reasons about the temporal behavior of streams, ensuring values are available when needed.
Given the undecidability of the problem, we implemented a syntactic over-approximation. 
We also extended the definition of well-formed specifications to detect temporal cycles that could lead to non-determinism.

Our evaluation, based on aerospace examples and worst-case benchmarks, demonstrates that the proposed type analysis is scalable and the syntactic fragment is sufficient for real-world applications.
Looking forward, we aim to refine the safe specification fragment through abstract interpretation, incorporating semantic reasoning into the analysis.
Additionally, we will explore more efficient monitor generation for parameterized specifications in real-time environments.


%
%

\bibliographystyle{splncs04}
\bibliography{references}

\appendix
\section{\rtlola AST-Syntax of an RTLola Expression}
\label{rtlola:syn:expr:full}
\begin{bnfgrammar}
	expr : Expressions ::= \syncAccess{id}{p} $\mid$ \holdAccess{id}{p}
	| \offsetAccess{id}{p}{off} $\mid$ \default{expr}{expr}
	| \aggr{id}{p}{dur}{f\_a}
	| \functionExp{f}{$\texttt{expr}^*$} $\mid$ \constant{const}
	| \parameterAccess{p\_idx} $\mid$ \tuple{expr*}
	;;
	const : Constants \in \mathbb{V}
	;;
	id : Stream identifiers ::= id\_in | id\_out
	;;
	id\_in : Input Stream identifiers \in \idsI
	;;
	id\_out : Output Stream identifiers \in \idsO
	;;
	p : Parameters \in \texttt{expr} \times \dots \times \texttt{expr}
	;;
	p\_idx : Parameter Indexs \in \mathbb{N}
	;;
	off : Offsets \in \mathbb{Z}
	;;
	dur : Durations \in \mathbb{R}
	;;
	f : Functions \in \svalue^* \rightarrow \svalue
	;;
	f\_a : Aggregation functions \in (\mathbb{R}, \mathbb{V})^* \rightarrow \svalue
\end{bnfgrammar}

\newpage
\section{Relational Semantics}
\subsection{Semantic Inference Rules for Expressions}
\label{fig:sema:expressions}
	\begin{mathpar}
    	\inferrule[eval-expr-val]{ \val \in \mathbb{V} }{\evalTo{\constant{\val}}{\val}}\qquad
        \inferrule[eval-expr-syn]{\evalTo{\mathit{expr}}{\mathit{inst'}}\\\worldaccthree{\sid}{\mathit{inst'}}{t} = v}{\evalTo{\syncAccess{\sid}{\mathit{expr}}}{v}}\qquad\\
        \inferrule[eval-expr-function]{\evalTo{p_1}{v_1} \dots \evalTo{p_n}{v_n}\\ f(v_1, \dots, v_n) = v_{\mathit{res}}}{\evalTo{\functionExp{f}{p_1,\dots,p_n}}{v_{\mathit{res}}}}\qquad \\
        \inferrule[eval-expr-off]{\evalTo{\mathit{expr}}{\mathit{inst'}} \\ \mathit{prefix} = \prefix(\world,\sid, \mathit{inst'}, t) \\|\mathit{prefix}| > \mathit{off} \\ \mathit{prefix}[|\mathit{prefix}| - \mathit{off}] = (t', v)}{\evalTo{\offsetAccess{\sid}{\mathit{expr}}{\mathit{off}}}{v}}\qquad\\
        \inferrule[eval-expr-off-dft]{\evalTo{\mathit{expr}}{\mathit{inst'}} \\ \mathit{prefix} = \prefix(\world,\sid, \mathit{inst'}, t) \\|\mathit{prefix}| \leq \mathit{off}}{\evalTo{\offsetAccess{\sid}{\mathit{expr}}{\mathit{off}}}{\bot}}\qquad\\
        \inferrule[eval-expr-hold]{\evalTo{\mathit{expr}}{\mathit{inst'}} \\ \mathit{prefix} = \prefix(\world,\sid, \mathit{inst'}, t) \\ \mathit{prefix} = \mathit{prefix}' \cdot (t',v)}{\evalTo{\holdAccess{\sid}{\mathit{expr}}}{v}}\qquad\\
        \inferrule[eval-expr-hold-dft]{\evalTo{\mathit{expr}}{\mathit{inst'}} \\ \mathit{prefix} = \prefix(\world,\sid, \mathit{inst'}, t)\\\mathit{prefix} = \epsilon}{\evalTo{\holdAccess{\sid}{\mathit{expr}}}{\bot}}\qquad\\
        \inferrule[eval-expr-no-default]{\evalTo{expr}{v} \\ v \in \mathbb{V}}{\evalTo{\default{\mathit{expr}}{\mathit{dft}}}{v}}\qquad
        \inferrule[eval-expr-take-default]{\evalTo{expr}{\bot} \\ \evalTo{\mathit{dft}}{v}}{\evalTo{\default{\mathit{expr}}{\mathit{dft}}}{v}}\qquad\\
        \inferrule[eval-expr-parameter-access]{\instid(p\_idx) = v}{\evalTo{\parameterAccess{p\_idx} }{v}}\qquad\\
        \inferrule[eval-expr-aggregate]{\evalTo{\expr}{\mathit{inst'}} \\ \mathit{prefix} = \prefix(\world,\sid, \mathit{inst'}, t) \\ window = \window(\world, prefix, \worldaccone{t} - dur) \\ f_a(window) = v}{\evalTo{\aggr{\sid}{\expr}{\mathit{dur}}{f_a} }{v}}\qquad\\
        \inferrule[eval-expr-tuple]{\evalTo{\mathit{expr}_1}{v_1} \dots \evalTo{\mathit{expr}_n}{v_n} \\ (v_1, \dots, v_n) = n}{\evalTo{\tuple{\mathit{expr}_1, \dots, \mathit{expr}_n}}{v}}\qquad\\
    \end{mathpar}

\subsection{Semantic Inference Rules for Streams}
\label{fig:sema:instances:complete}
	\begin{mathpar}
    	\inferrule[eval-instance-value]{\alive(\world, \sid, \instid, t) = \tstart\\ \tstart \neq \bot \\\pacing{t}{\pac}{\tstart} \\ v \in \mathbb{V} \\ \evalTo{\mathit{when}}{true}\\\evalTo{\mathit{with}}{v}}{\evalTo{\sOutput{\sid}{\_}{\_}{\sEval{\pac}{\mathit{when}}{\mathit{with}}}{\_}}{v}}\qquad\\
    	\inferrule[eval-instance-false-eval-when]{\alive(\world, \sid, \instid, t) = \tstart\\ \tstart \neq \bot \\\pacing{t}{\pac}{\tstart} \\ v \in \mathbb{V} \\ \evalTo{\mathit{when}}{false}}{\evalTo{\sOutput{\sid}{\_}{\_}{\sEval{\pac}{\mathit{when}}{\mathit{with}}}{\_}}{\bot}}\qquad\\
    	\inferrule[eval-instance-false-pacing]{\alive(\world, \sid, \instid, t) = \tstart\\ \tstart \neq \bot \\ \neg \pacing{t}{\pac}{\tstart}}{\evalTo{\sOutput{\sid}{\_}{\_}{\sEval{\pac}{\mathit{when}}{\mathit{with}}}{\_}}{\bot}}\qquad\\
    	\inferrule[eval-instance-not-alive]{\alive(\world, \sid, \instid, t) = \bot}{\evalTo{\sOutput{\sid}{\_}{\_}{\sEval{\pac}{\mathit{when}}{\mathit{with}}}{\_}}{\bot}}\qquad
	\end{mathpar}

\subsection{Definition of IsSpawned and IsClosed}
\label{def:is_spawned}
\begin{definition}[IsSpawned \& IsClosed]
	Given an evaluation model $\world \in \World$, a stream identification $\sid \in \sr$, an instance identification $\instid \in \instanceId$, and a timestamp $t \in \stime$, the instance $\instid$ of stream $\sid$ is spawned at $t$ iff
	\[
	\begin{array}{l}
		\isSpawned : \World \times \sr \times \instanceId \times \stime \rightarrow \mathbb{B}\\
		\isSpawned(\world, \sid, \instid, t) \\~~:=\pacing{\tstart}{\spawnPac{\sid}}{0} \wedge \evalToAt{t}{\instid}{\spawnWhen{\sid}}{\mathit{true}}\\~~~~\wedge \evalToAt{t}{\instid}{\spawnWith{\sid}}{\instid}\\
	\end{array}
	\]
	and is closed iff
	\[
	\begin{array}{l}
		\isClosed : \World \times \sr \times \instanceId \times \stime \times \stime \rightarrow \mathbb{B}\\
		\isClosed(\world, \sid, \instid, t, t_{\instid \mathit{start}})\\~~:=\pacing{t}{\closePac{\sid}}{t_{\instid \mathit{start}}} \rightarrow \evalToAt{t}{\instid}{\closeWhen{\sid}}{\mathit{true}}\\
	\end{array}
	\]
\end{definition}

\section{Type System Soundness Proofs}
\label{app:proofs}
\subsection{Proof of the Failable Semantics Lemma}
\label{proof:failable}
\begin{proof}[Failable Semantics]
	Proof by structural induction over the syntax of \rtlola specifications.
	As a base cases regard the constant and parameter access, by the assumptions made above these rules cannot fail.
	For all other expressions we argue about their semantic definition.
	For all their subexpressions we directly apply the induction hypothesis.
	The value type analysis ensures that a hold an offset operator are always equipped with a default operator and hence can never produce $\bot$
	The default expression itself can never produce $\bot$.
	For the function and aggregation expression, the value type analysis enforces that, if the function could return $\bot$, then it is again equipped with a default value.
	For the tuple, the induction hypothesis holds.
	As a result only the synchronous access could evaluate to $\bot$ if $\worldaccthree{\sid}{\mathit{inst'}}{t} = \bot$.
	It is easy to see the semantics can only fail if an expression produces $\bot$ when it shouldn't and as the synchronous access is the only expression that can produce $\bot$ the lemma follows.
\end{proof}

\subsection{Proof of the Abstract Safety Theorem}
\label{proof:abstract-safety}
\begin{proof}[Abstract Safety]
	Proof by contradiction.
	Not $safe(\varphi)$, implies that there is an input sequence without a coherent world.
	Due to \Cref{lemma:failable_semantics} this is due to a synchronous access producing $\bot$.
	Yet, for all synchronous accesses the condition of the implication holds.
	Hence, it the condition must also hold for the synchronous access that produced $\bot$.
	This is a contradiction, as the condition enforces, that the accessed stream is active and produced a value.
\end{proof}

\subsection{Proof of the Typechecked-Safety Theorem}
\label{proof:typechecked}
\begin{proof}
	The theorem follows directly from the type inference rules for the synchronous access and the output stream.
	Each conjunction of the $\mathit{active}$ predicate of on the left side of the implication implies the corresponding conjunct of the $\mathit{active}$ predicate of the right side of the implication.
	To make an example, the \emph[Sem-Output] rule requires that the semantic type of the spawn condition of the stream is more concrete (or even equal) to the semantic spawn type inferred for the eval clause.
	The \emph{Sem-Sync-Output} rule applied to synchronous accesses infers a semantic spawn type that is equal or more concrete than the one of the accessed stream.
	Through this chain of requirements it follows from the definition of the semantic type lattice, that the \emph{when} expression of the accessing streams spawn clause must always imply the \emph{when} expression of the accessed streams spawn clause.
	From this one can conclude, that the $\mathit{AliveSince}$ predicate of the accessing stream evaluates to a point in time equal or after the $\mathit{AliveSince}$ predicate of the accessed stream.
	The same argument can be conducted analogously for the other conjuncts of the $\mathit{active}$ predicate.
\end{proof}

\section{RTLolas Value Type System}
\label{appendix:value_types}
\begin{definition}[The value type lattice]
    The set of primitive value types is defined as:
    \[
        \vtPrim:= \{Bool, String, UInt(x), Int(x), Float(y) \mid x \in \{8, 16, 32, 64\}, y \in \{32, 64\}\}    
    \]
    Each of these primitive types can be wrapped in an option to indicate the absence of a value.
    \[
        \VT := \vtPrim \cup \{\optional{\vt} \mid \vt \in \vtPrim\} \cup \{(\vt_1 \times ... \times \vt_n) \mid \vt_1,...,\vt_n \in \VT\} \cup \{\top, \bot\}
    \]
    The relation $\vtCon$ is minimal with regard to the following equations:
    \begin{alignat*}{2}
        \vt &\vtCon \top && \Longleftrightarrow \vt \in \VT\\
        \bot &\vtCon \vt && \Longleftrightarrow \vt \in \VT\\
        Bool &\vtCon Bool && \Longleftrightarrow true\\
        String &\vtCon String && \Longleftrightarrow true\\
        UInt(x_1) &\vtCon UInt(x_2) && \Longleftrightarrow x_1 \geq x_2\\
        Int(x_1) &\vtCon Int(x_2) && \Longleftrightarrow x_1 \geq x_2\\
        Float(y_1) &\vtCon Float(y_2) && \Longleftrightarrow y_1 \geq y_2\\
        \optional{\vt_1} &\vtCon \optional{\vt_2} && \Longleftrightarrow \vt_1 \vtCon \vt_2
    \end{alignat*}
\end{definition}
Note that the value types can easily be extended to fixed-size arrays like strings or fixed-arity compound types like tuples or structs.
To accommodate explicit type annotations we assume that the following two functions are given:
\begin{align*}
	\alpha &: \ids \rightarrow \VT\\
	\rho &: \ids \rightarrow \VT*
\end{align*}
Where $\alpha$ assigns a value type to each stream (or $\top$ if not annotated) and $\rho$ assigns a value type to each parameter of the given stream.

\begin{figure}
	\begin{mathpar}
		\inferrule[value-constant]
		{ }
		{\vt \types{} \constant{x}}
		
		\inferrule[value-sync-out]
		{
			\sid \in \idsO\\
			\vt' \types{} \sOutput{\sid}{p\_def}{s}{e}{c}\\
			\vt \vtCon \vt'\\
			\parameterTypes{\sid} \types{} p
		}
		{\vt \types{} \syncAccess{\sid}{p}}
		
		\inferrule[value-sync-in]
		{
			\sid \in \idsI\\
			\vt' \types{} \sInput{\sid}\\
			\vt \vtCon \vt'\\
			\parameterTypes{\sid} \types{} p
		}
		{\vt \types{} \syncAccess{\sid}{p}}
		
		\inferrule[value-offset-out]
		{
			\sid \in \idsO\\
			\vt' \types{} \sOutput{\sid}{p\_def}{s}{e}{c}\\
			\vt \vtCon \vt'\\
			\parameterTypes{\sid}\types{} p
		}
		{\optional{\vt} \types{} \offsetAccess{\sid}{p}{o}}
		
		\inferrule[value-offset-in]
		{
			\sid \in \idsI\\
			\vt' \types{} \sInput{\sid}\\
			\vt \vtCon \vt'\\
			\parameterTypes{\sid} \types{} p
		}
		{\optional{\vt} \types{} \offsetAccess{\sid}{p}{o}}
		
		\inferrule[value-hold-output]
		{
			\sid \in \idsO\\
			\vt'\types{} \sOutput{\sid}{p\_def}{s}{e}{c}\\
			\vt \vtCon \vt'\\
			\parameterTypes{\sid} \types{} p
		}
		{\optional{\vt} \types{} \holdAccess{\sid}{p}}
		
		\inferrule[value-hold-input]
		{
			\sid \in \idsI\\
			\vt' \types{} \sInput{\sid}\\
			\vt \vtCon \vt'\\
			\parameterTypes{\sid} \types{} p
		}
		{\optional{\vt} \types{} \holdAccess{\sid}{p}}
		
		\inferrule[value-aggregation-output]
		{
			\sid \in \idsO\\
			\vt \types{} \sOutput{\sid}{p\_def}{s}{e}{c}\\
			f: \vt_f \times ... \times \vt_f \rightarrow \vt'\\
			\vt_f \vtCon \vt\\
			\parameterTypes{\sid} \types{} p
		}
		{\vt' \types{} \aggr{\sid}{p}{d}{f}}
		
		\inferrule[value-aggregation-input]
		{
			\sid \in \idsI\\
			\vt \types{} \sInput{\sid}\\
			f: \vt_f \times ... \times \vt_f \rightarrow \vt'\\
			\vt_f \vtCon \vt\\
			\parameterTypes{\sid} \types{} p
		}
		{\vt' \types{} \aggr{\sid}{p}{d}{f}}
		
		\inferrule[value-default]
		{
			\optional{v_1} \types{} exp_1\\
			\vt_2 \types{} exp_2\\
			\vt \vtCon \vt_1\\
			\vt \vtCon \vt_2
		}
		{\vt \types{} \default{exp_1}{exp_2}}
		
		\inferrule[value-parameter]
		{
			\parameterTypes{\sid}[idx] = \vt
		}
		{\vt \types{\sid} \parameterAccess{idx}}
		
		\inferrule[value-function]
		{
			f: \vt_1 \times ... \times \vt_n \rightarrow \vt\\
			\vt_1' \types{} exp_1\\
			\vt_1 \vtCon \vt_1'\\
			\dots\\
			\vt_n' \types{} exp_n\\
			\vt_n \vtCon \vt_n'\\
		}
		{\vt \types{} \functionExp{f}{expr_1, ..., expr_n}}
		
		\inferrule[value-tuple]
		{
			\vt_1 \types{} expr_1\\
			\dots\\
			\vt_n \types{} expr_n
		}
		{(\vt_1 \times ... \times \vt_n) \types{} \tuple{expr_1, ..., expr_n}}
	
	\end{mathpar}
    \caption{RTLola value type inference rules}
\end{figure}

\begin{figure}
	\begin{mathpar}
   		\inferrule[value-output]
   		{
   			s = \sSpawn{pt_s}{e_s^f}{e_s}\\
   			\vt_s \types{} \sSpawn{pt}{e_s^f}{e_s}\\
   			\parameterTypes{\sid} \vtCon \vt_s\\
   			e = \sEval{pt_e}{e_e^f}{e_e}\\
   			\vt_e \types{} \sEval{pt_e}{e_e^f}{e_e}\\
   			\vt \vtCon \vt_e\\
   			c = \sClose{pt_c}{e_c^f}\\
   			() \types{} \sClose{pt_c}{e_c^f}
   		}
   		{\vt \types{} \sOutput{\sid}{p\_def}{s}{e}{c}}
   		
   		\inferrule[value-input]
   		{
   			\streamTypes{\sid} = \vt
   		}
   		{\vt \types{} \sInput{\sid}}
   		
   		\inferrule[value-spawn]
   		{
   			\bool \types{} e_s^f\\
   			\vt \types{} e_s
   		}
   		{\vt \types{} \sSpawn{pt_s}{e_s^f}{e_s}}
   		
   		\inferrule[value-eval]
   		{
   			\bool \types{} e_e^f\\
   			\vt \types{} e_e
   		}
   		{\vt \types{} \sEval{pt_e}{e_e^f}{e_e}}
   		
   		\inferrule[value-close]
   		{
   			\bool \types{} e_c^f
   		}
   		{() \types{} \sClose{pt_c}{e_c^f}}
	\end{mathpar}
    \caption{Value type rules for streams}
\end{figure}

\newpage
\section{RTLolas Pacing Type System}
\label{appendix:pacing_types}
\begin{figure}[H]
	\begin{mathpar}
		\inferrule[pacing-input]
		{ }
		{(\top, (Event(id), \top) \types{} \sInput{\sid}}
		
		\inferrule[pacing-output]
		{
			(\top, \pt_e^s, \top) \types{} \sSpawn{\pt_s'}{e_s^f}{e_s}\\
			(\pt_s^e, \pt_e^e, \pt_c^e) \types{} \sEval{\pt_e'}{e_e^f}{e_e}\\
			(\pt_s^c, \pt_e^c, \pt_c^c) \types{} \sClose{\pt_c'}{e_c^f}\\
			\pt_s^c \in \{\top, \pt_e^s\}\\
			\pt_c^c \in \{\top, \pt_e^c\}\\
			\pt_e^e = \lf{x} \rightarrow \pt_e^s = \pt_s^e \land \pt_e^c = \pt_c^e \\
			\pt_e^c = \lf{x} \rightarrow \pt_e^s = \pt_s^c \land \pt_e^c = \pt_c^c\\
			\pt_e^s \ptCon \pt_s^e\\
			\pt_c^e \ptCon \pt_e^c\\
		}
		{(\pt_e^s, \pt_e^e, \pt_e^c) \types{} \sOutput{\sid}{p\_def}{\sSpawn{\pt_s'}{e_s^f}{e_s}}{\sEval{\pt_e'}{e_e^f}{e_e}}{\sClose{\pt_c'}{e_c^f}}}
		
		\inferrule[pacing-spawn]
		{
			\ptt_f \types{} e_s^f\\
			\ptt' \ptCon \ptt_f\\
			\ptt_e \types{} e_s\\
			\ptt' \ptCon \ptt_e\\
			\pt_e \ptCon \pt_e'
		}
		{(\pt_s', \pt_e, \pt_c') \types{} \sSpawn{\pt_e}{e_s^f}{e_s}}
		
		\inferrule[pacing-eval]
		{
			\ptt_f \types{} e_e^f\\
			\ptt' \ptCon \ptt_f\\
			\ptt_e \types{} e_e\\
			\ptt' \ptCon \ptt_e\\
			\pt_e \ptCon \pt_e'
		}
		{(\pt_s', \pt_e, \pt_c') \types{} \sEval{\pt_e}{e_e^f}{e_e}}
		
		\inferrule[pacing-close]
		{
			\ptt_f \types{} e_c^f\\
			\ptt' \ptCon \ptt_f\\
			\pt_e' \ptCon \pt_e
		}
		{(\pt_s', \pt_e, \pt_c') \types{} \sClose{\pt_e}{e_c^f}}
	\end{mathpar}
    \caption{Inference rules for the pacing type analysis of streams.}
    \label{fig:pacing_streams}
\end{figure}

\begin{figure}[H]
	\begin{mathpar}
		\inferrule[pacing-constant]
		{ }
		{(\top, \top, \top) \types{} \constant{c}}
		
		\inferrule[pacing-sync-input]
		{
			\ptt \types{} \sInput{\sid}\\
		}
		{
			\ptt \types{} \syncAccess{\sid}{p}
		}
		
		\inferrule[pacing-sync-output]
		{
			\ptt' \types{} \sOutput{\sid}{n_p}{s}{e}{c}\\
			\ptt \ptCon \ptt'\\
			\ptt_1 \types{} e_1\\
			\ptt \ptCon \ptt_1\\
			\dots\\
			\ptt_{n_p} \types{} e_{n_p}\\
			\ptt \ptCon \ptt_{n_p}
		}
		{
			\ptt \types{} \syncAccess{\sid}{e_1,...,e_{n_p}}
		}
		
		\inferrule[pacing-hold]
		{
			\ptt_1 \types{} e_1\\
			\ptt \ptCon \ptt_1\\
			\dots\\
			\ptt_{n_p} \types{} e_{n_p}\\
			\ptt \ptCon \ptt_{n_p}
		}
		{\ptt \types{} \holdAccess{\sid}{e_1, ..., e_{n_p}}}
		
		\inferrule[pacing-offset-input]
		{
			\ptt \types{} \sInput{\sid}\\
		}
		{
			\ptt \types{} \offsetAccess{\sid}{p}{o}
		}
		
		\inferrule[pacing-offset-output]
		{
			\ptt' \types{} \sOutput{\sid}{n_p}{s}{e}{c}\\
			\ptt \ptCon \ptt'\\
			\ptt_1 \types{} e_1\\
			\ptt \ptCon \ptt_1\\
			\dots\\
			\ptt_{n_p} \types{} e_{n_p}\\
			\ptt \ptCon \ptt_{n_p}
		}
		{
			\ptt \types{} \offsetAccess{\sid}{e_1, ..., e_{n_p}}{o}
		}
		
		\inferrule[pacing-default]
		{   
			\ptt_1 \types{} e_1\\
            \ptt_2 \types{} e_2\\
            \ptt \ptCon \ptt_1\\
            \ptt \ptCon \ptt_2
        }
		{\ptt \types{} \default{e_1}{e_2}}
		
		\inferrule[pacing-parameter]
		{ }
		{(\top, \top, \top) \types{} \parameterAccess{i}}
		
		\inferrule[pacing-function]
		{
			\ptt_1 \types{} e_1\\
			\ptt \ptCon \ptt_1\\
			\dots\\
			\ptt_n \types{} e_n\\
			\ptt \ptCon \ptt_n
		}
		{\ptt \types{} \functionExp{f}{e_1,...e_n}}
		
		\inferrule[pacing-aggregation]
		{
			\ptt \ptCon \textit{Periodic}\\
			\ptt_1 \types{} e_1\\
			\ptt \ptCon \ptt_1\\
			\dots\\
			\ptt_{n_p} \types{} e_{n_p}\\
			\ptt \ptCon \ptt_{n_p}
		}
		{\ptt \types{} \aggr{\sid}{e_1, ..., e_{n_p}}{d}{f}}
		
		\inferrule[pacing-tuple]
		{
			\ptt_1 \types{} e_1\\
			\ptt \ptCon \ptt_1\\
			\dots\\
			\ptt_{n_p} \types{} e_n\\
			\ptt \ptCon \ptt_n
		}
		{\ptt \types{} \tuple{e_1, ..., e_n}}
	\end{mathpar}
	\caption{Inference rules for the pacing type analysis of expressions.}
	\label{fig:pacing_expr}
\end{figure}

\newpage
\section{RTLolas Semantic Type System}
\label{appendix:semantic_types}
\begin{figure}[H]
    \begin{mathpar}
    	\inferrule[Sem-Constant]
    	{ }
    	{(\top, \top, \bot) \types{} \constant{v}}
    
    	\inferrule[Sem-Sync-In]
    	{
    		\sid \in \idsI
    	}
    	{(\top, \top, \bot) \types{} \syncAccess{\sid}{p}}
    	
    	\inferrule[Sem-Sync-Out]
    	{
    		\sOutput{\sid}{n_p}{s}{e}{c} \in \outputs\\
    	    \sOutput{\sid'}{n_p'}{s'}{e'}{c'} \in \outputs\\
            \forall 1 \leq k \leq n.\ e_k = \parameterAccess{i} \land k \paramEq{\sid}{\sid'} i\\
            \stt \types{} \sOutput{\sid}{n_p}{s}{e}{c}\\
            \stt'\stCon \stt
        }
    	{\stt' \types{\sid'} \syncAccess{\sid}{(e_1,...,e_n)}}
    	
    	\inferrule[Sem-Offset-In]
    	{
    		\sid \in \idsI
    	}
    	{(\top, \top, \bot) \types{}  \offsetAccess{\sid}{p}{o}}
    	
    	\inferrule[Sem-Aggr]
    	{
    		\stt^1 \types{} e_1\\
    		\stt \stCon \stt^1\\
    		\dots\\
    		\stt^n \types{} e_n\\
    		\stt \stCon \stt^n\\
    	}
    	{\stt \types{} \aggr{id}{e_1, ..., e_n}{d}{f}}
    	
    	\inferrule[Sem-Offset-Out]
    	{
    		\sOutput{\sid}{n_p}{s}{e}{c} \in \outputs\\
    	    \sOutput{\sid'}{n_p'}{s'}{e'}{c'} \in \outputs\\
            \forall 1 \leq k \leq n.\ e_k = \parameterAccess{i} \land \exists 1 \leq l \leq n_p.\ i \paramEq{\sid}{\sid'} l\\
            \stt \types{} \sOutput{\sid}{n_p}{s}{e}{c}\\
            \stt'\stCon \stt
        }
    	{\stt' \types{\sid'} \offsetAccess{\sid}{p}{o}}
    	
    	\inferrule[Sem-Hold]
    	{
    		\stt^1 \types{} e_1\\
    		\stt \stCon \stt^1\\
    		\dots\\
    		\stt^n \types{} e_n\\
    		\stt \stCon \stt^n\\
    	}
    	{\stt \types{} \holdAccess{id}{e_1,...,e_n}}
    	
    	\inferrule[Sem-Parameter-Access]
    	{ }
    	{(\top, \top, \bot) \types{} \parameterAccess{i}}
    	
    	\inferrule[Sem-Default]
    	{
  		    \stt^1 \types{} e_1\\
  		    \stt^2 \types{} e_2\\
			\stt \stCon \stt^1\\
			\stt \stCon \stt^2
        }
    	{\stt \types{} \default{e_1}{e_2}}
    	
    	\inferrule[Sem-Function]
    	{
    		\stt^1 \types{} e_1\\
    		\stt \stCon \stt^1\\
    		\dots\\
    		\stt^n \types{} e_n\\
    		\stt \stCon \stt^n\\
    	}
    	{\stt \types{} \functionExp{f}{e_1,...,e_n}}
    	
    	\inferrule[Sem-Tuple]
    	{
    		\stt^1 \types{} e_1\\
    		\stt \stCon \stt^1\\
    		\dots\\
    		\stt^n \types{} e_n\\
    		\stt \stCon \stt^n\\
    	}
    	{\stt \types{} \tuple{e_1,...,e_n}}
    \end{mathpar}
	\caption{Inference rules for the semantic type analysis of expressions.}
    \label{fig:semantic_expr}
\end{figure}

\begin{figure}[H]
	\begin{mathpar}
		\inferrule[sem-output]
		{
		(\top, \st_e^s, \bot) \types{\sid} \sSpawn{p_s}{e_s^f}{e_s}\\
		\stt^e \types{\sid} \sEval{\pt_e}{e_e^f}{e_e}\\
		(\st_s^c, \st_e^c, \st_c^c) \types{\sid} \sClose{\pt_c}{e_c^f}\\
		\st_s^c \in \{\top, \st_e^s\}\\
		\st_c^c \in \{\bot, \st_e^c\}\\
		\pt_e = \lf{x} \rightarrow \st_e^s = \st_s^e \land \st_e^c \land \st_c^e\\
		\pt_c = \lf{x} \rightarrow \st_e^s = \st_s^c \land \st_e^c = \st_c^c\\
		\st_e^s \stCon \st_s^e\\
		\st_c^e \stCon \st_e^c
		}
		{(\st_e^s, \st_e^e, \st_e^c) \types{} \sOutput{\sid}{p\_def}{\sSpawn{p_s}{e_s^f}{e_s}}{\sEval{p_e}{e_e^f}{e_e}}{\sClose{p_c}{e_c^f}}}
		
		\inferrule[sem-input]
		{ }
		{(\top, \top, \bot) \types{} \sInput{\sid}}
		
		\inferrule[sem-spawn]
		{
		 \stt_f \types{\sid} e_s^f\\
		 \stt' \types{\sid} e_s\\
		 \stt \stCon \stt_f\\
		 \stt \stCon \stt'\\
		 Expr(e_s^f) \stCon \st_e
		}
		{(\st_s, Expr(e_s^f), \st_c) \types{\sid} \sSpawn{p\_def}{e_s^f}{e_s}}
		
		\inferrule[sem-eval]
		{
		 \stt_f \types{\sid} e_e^f\\
		 \stt' \types{\sid} e_e\\
		 \stt \stCon \stt_f\\
		 \stt \stCon \stt'\\
		 Expr(e_e^f) \stCon \st_e
		}
		{(\st_s, Expr(e_e^f), \st_c) \types{\sid} \sEval{p\_def}{e_e^f}{e_e}}
		
		\inferrule[sem-close]
		{
		 \stt \types{\sid} e_c^f\\
		 \st_c \stCon Expr(e_c^f)
		}
		{(\st_s, Expr(e_c^f), \st_c) \types{\sid} \sClose{p\_def}{e_c^f}}
	\end{mathpar}
    \caption{Inference rules for the semantic type analysis of streams.}
    \label{fig:semantic_streams}
\end{figure}

\section{Benchmarking Specifications}
\label{appendix:benchmarks}

\subsection{Intruder}
\label{ex:drone_distance}
The following specification monitors the distance from one drone to multiple other drones.
It is the original version of the simplified specification of \Cref{sec:intro:example} as presented in \cite{DBLP:conf/fm/BaumeisterFKS24}.
\begin{lstlisting}
import math
input lat: Float
input lon: Float
input intruder_id: UInt
input intruder_lat: Float
input intruder_lon: Float

output intruder_pos(id)
    spawn with intruder_id
    eval when id = intruder_id with  (intruder_lat, intruder_lon)
    close @true@ when stale(id).hold(or: false)
output distance(id)
    spawn with intruder_id
    eval @((intruder_id && intruder_lat && intruder_lon) || (lat &&lon))@
    with sqrt((intruder_pos(id).hold().0.defaults(to: 0.0) - lat.hold(or: 0.0))**2.0 + (intruder_pos(id).hold().1.defaults(to: 0.0) - lon.hold(or: 0.0))**2.0)
    close @true@ when stale(id).hold(or: false)
output closer(id)
    spawn with intruder_id
    eval with distance(id).offset(by: -1).defaults(to: distance(id)) >= distance(id)
    close @true@ when stale(id).hold(or: false)
output stale(id)
    spawn with intruder_id
    eval @10s@ with intruder_pos(id).aggregate(over: 10s, using: count) = 0
    close @true@ when stale(id).hold(or: false)

trigger(id)
    spawn when distance(intruder_id).hold(or: 1.0) < 0.1 with intruder_id
    eval @1Hz@ when closer(id).aggregate(over_exactly: 5s, using: forall).defaults(to: false) with "Intruder {{}} detected".format(id)
    close @true@ when stale(id).hold(or: false)
\end{lstlisting}
%

\subsection{Waypoint Mission}
\label{ex:waypoints}
The example demonstrates how a waypoint mission with arbitrary many waypoints can be realized in \rtlola.
For that streams are parameterized over the x and y coordinates of the waypoints, such that a single stream instance is responsible for a single waypoint.
\begin{lstlisting}
	input pos: (Float, Float)
	input waypoint: (Float, Float)

	output waypoint_distance(wx: Float, wy: Float)
  		spawn with (waypoint.0, waypoint.1)
  		eval with sqrt((pos.0 - wx)**2.0 + (pos.1 - wy)**2.0)
  		close when waypoint_reached(wx, wy)

	output waypoint_approaching(wx: Float, wy: Float)
  		spawn with (waypoint.0, waypoint.1)
  		eval with waypoint_distance(wx, wy) <= waypoint_distance(wx, wy).offset(by: -1).defaults(to: 0.0)
  		close when waypoint_reached(wx, wy)

	output waypoint_reached(wx: Float, wy: Float)
  		spawn with (waypoint.0, waypoint.1)
  		eval with waypoint_distance(wx, wy) < 5.0
  		close when waypoint_reached(wx, wy)
\end{lstlisting}
Assume the input stream \lstinline{pos} contains the current position of a drone as an $(x, y)$ coordinate, while the input stream \lstinline{waypoint} contains the $(x, y)$ coordinates of the next waypoint to be monitored.
Each output stream is parameterized over these waypoint coordinates through identical spawn clauses.
Each instance is closed when the corresponding waypoint is reached, implemented through the \lstinline{waypoint\_reached} output stream accesses in the close clauses.
For that, \lstinline{waypoint\_distance} stream calculates the euclidian distance between the current position of the drone and the waypoint.
Once, this distance is below a certain threshold, the corresponding instance of the \lstinline{waypoint\_reached} output stream evaluates to true.
Additionally, the \lstinline{waypoint\_approaching} stream monitors, again for each waypoint, if the drone is getting closer the waypoints through a temporal offset.
Concretely, it compares the current distance to the waypoint with the previous one and asserts that the latter one is higher.
The access of the previous value is implemented through the \lstinline{offset} access accompanied by a default value, as a previous value does not exist in the first step of the monitor execution.

\subsection{Watchdog}
\label{ex:watchdog}
The example realizes a monitor that supervises the reactivity of electronic
control units (ECUs) inside an aircraft.
\begin{lstlisting}
	input ping: Int
	input pong: Int

	output pong_of_node(node: Int)
    	spawn with pong
    	eval with pong = node

	output is_alive(node: Int)
    	spawn with ping
    	eval @1min with pong_of_node(node).aggregate(over: 1min, using: exists)
    	close when is_alive(node) = is_alive(node)
\end{lstlisting}
The values of the two input streams \streamName{ping} and \streamName{pong} represent the numeric identifier of the ECU.
If a value on the \streamName{ping} stream is received, an alive request has been issued to that ECU that needs to be answered by this unit realized with the \streamName{pong} stream.
The parameterized output stream \streamName{pong\_of\_node} groups the pong responses by spawning a new instance for each ECU that evaluates only if pong from the same ECU is received.
The \streamName{is\_alive} output stream implements a watchdog for each ping request.
This watchdog checks if the response was received within a defined period. 
Again, we need a parametrized output stream that spawns with each ping.
Each instance is then evaluated after one minute, stating with the creation of the instance.
The evaluation aggregates over the pong responses of this unit to check whether the unit responded within the last minute.
A final close declaration deletes the instance from the memory after the first evaluation to prepare the monitor for the next ping of this ECU.

\subsection{RCC}
\label{ex:rcc}
The example is taken from \cite{volostream} and monitors the remote control unit of an unmanned aircraft.
\begin{lstlisting}
    input seq_number : Int64

    /// Helper Streams
    output lost_connection_to_master @true@ := false
    output switch_to_secondary @true@ := false
    output both_rc_disconnected @true@ := false

    /// Property 1: Log message increment
    output valid_seq_number := seq_number = seq_number.offset(by: -1, or: -1) + 1
    /// Property 7: RC fallback test
    output main_fallback_valid_dyn
        spawn when lost_connection_to_master
        close when switch_to_secondary || both_rc_disconnected
        eval @200ms@ with false
    output main_fallback_valid @true@ := main_fallback_valid_dyn.hold(or: true)
\end{lstlisting}

\subsection{FFD}
\label{ex:ffd}
The example is taken from \cite{volostream} and monitors the flight phases of an unmanned aircraft.
\begin{lstlisting}
import math

constant ROTOR_1 : UInt8 := 1
constant ROTOR_2 : UInt8 := 2
constant EPSILON_RPM_ON : Float64 := 1.0
constant EPSILON_RPM : Float64 := 0.5


input rpm : Int64
input src : UInt8

output rpm_1 eval when src == ROTOR_1 with abs(cast<Int64, Float64>(rpm))
output rpm_2 eval when src == ROTOR_2 with abs(cast<Int64, Float64>(rpm))

output rpm_on_check @rpm@ := if rpm_1.hold(or: 0.0) + rpm_2.hold(or: 0.0) / 2.0 > EPSILON_RPM_ON then 1.0 else 0.0
output rpm_on @1s@ := rpm_on_check.aggregate(over: 1s, using: avg).defaults(to: 0.0)  > EPSILON_RPM
output take_off @1s@ := false
output landed @1s@ := false
output rpm_in_air @1s@ := false


output phase_1 := !take_off && !landed && rpm_on && !rpm_in_air
\end{lstlisting}

\subsection{Geofence}
\label{ex:geofence}
This example demonstrates a geofence with four points spanning polygone an autonomous aircraft is not allowed to leave.
\begin{lstlisting}
import math
/// Position
input gps__latitude :Float64
input gps__longitude :Float64
input gps__height :Float64
/// Velocity
input gps__speed_x :Float64
input gps__speed_y :Float64
input gps__speed_z :Float64
/// Accuracy Checks
input gps__horizontal_accuracy :Float64
input gps__vertical_accuracy :Float64

///****************///
///****GEOFENCE****///
///****************///
constant c_epsilon : Float64 := 0.0000001
constant close_to_geofence : Float64 := 3.00000000000000000000
constant min_time_to : Float64 := 5.0
constant inside_g1: Bool := true
constant inside_g2: Bool := true
constant inside_g3: Bool := true

// Computes the vehicle line
output lat_in_rad eval when gps_condition with gps__latitude  * 3.14159265359 / 180.0
output lon_in_rad eval when gps_condition with gps__longitude  * 3.14159265359 / 180.0
output filtered_height eval when gps_condition with gps__height
output velocity_xy eval when gps_condition with sqrt(gps__speed_x**2.0 + gps__speed_y**2.0)
output initial_height eval when gps_condition with initial_height.offset(by: -1).defaults(to: gps__height)
output x eval when gps_condition with lat_in_rad
output y eval when gps_condition with lon_in_rad
output delta_x eval when gps_condition with x - x.last(or: x)
output delta_y eval when gps_condition with y - y.last(or: x)
output isFnc eval when gps_condition with x != x.last(or: x)
output gradient eval when gps_condition with if isFnc then delta_y / delta_x else 0.0
output y_intercept eval when gps_condition with y - (gradient * x)
output dstToPnt eval when gps_condition with sqrt(delta_x**2.0 + delta_y**2.0)
output orientation_x eval when gps_condition with if abs(delta_x) < c_epsilon then false else delta_x < 0.0
output orientation_y eval when gps_condition with if abs(delta_y) < c_epsilon then false else delta_y < 0.0

constant height_violation : Bool := false

///******************///
///****Conditions****///
///******************///
output condition_horizontal_accuracy := ((gps__horizontal_accuracy) < 0.01)
output condition_vertical_accuracy := ((gps__vertical_accuracy) < 0.01)
output gps_condition := condition_horizontal_accuracy && condition_vertical_accuracy

// Polygonline g0: (76.07671638200421, 17.139048054333795) to (73.89655156587213, 52.81763088322348)
constant gradient_g1_0 :Float64 := -0.06110570104726112
constant y_intercept_g1_0 :Float64 := 21.787769142229816
output g1_0_is_intersecting eval when gps_condition with gradient != gradient_g1_0 or y_intercept = y_intercept_g1_0
output g1_0_x eval when gps_condition with if isFnc and g1_0_is_intersecting then (y_intercept - y_intercept_g1_0) / (gradient_g1_0 - gradient) else x
output g1_0_y eval when gps_condition with gradient_g1_0 * g1_0_x + y_intercept_g1_0
output g1_0_x_deg eval when gps_condition with g1_0_x * 180.0 / 3.14159265359
output g1_0_y_deg eval when gps_condition with g1_0_y * 180.0 / 3.14159265359
output g1_0_orientation eval when gps_condition with ( (if abs(g1_0_x - x.last(or: x)) < c_epsilon then false else g1_0_x -  x.last(or: x) < 0.0 ) = orientation_x ) and ( (if abs(g1_0_y -  y.last(or: y)) < c_epsilon then false else g1_0_y - y.last(or: y) < 0.0 ) = orientation_y)
output g1_0_violated_cnd eval when gps_condition with (!isFnc or g1_0_is_intersecting) and g1_0_orientation and (g1_0_x < 73.89655156587213 and g1_0_x > 17.139048054333795)and (g1_0_y < 52.81763088322348 and g1_0_y > 17.139048054333795)
output g1_0_is_violated eval when gps_condition with g1_0_violated_cnd and (sqrt((x.last(or: x) - g1_0_x)**2.0 + (y.last(or: y) - g1_0_y)**2.0) <= dstToPnt)
output g1_0_distance_to eval when gps_condition with if g1_0_violated_cnd then 6371.0 * 2.0 * arcsin(sqrt((sin((x - g1_0_x) / 2.0) ** 2.0) + cos(x) * cos(g1_0_x) * (sin((y - g1_0_y) / 2.0) ** 2.0))) * 1000.0 else 1000.0
output g1_0_time_to eval when gps_condition with g1_0_distance_to / velocity_xy
trigger if gps_condition then g1_0_is_violated.hold(or: false) else false "Violation: g1_0!"

// Polygonline g1: (56.286685329564456, 37.16234411824784) to (55.16368257693344, 76.79656929555932)
constant gradient_g1_1 :Float64 := -0.02833416693796948
constant y_intercept_g1_1 :Float64 := 38.75718045676067
output g1_1_is_intersecting eval when gps_condition with gradient != gradient_g1_1 or y_intercept = y_intercept_g1_1
output g1_1_x eval when gps_condition with if isFnc and g1_1_is_intersecting then (y_intercept - y_intercept_g1_1) / (gradient_g1_1 - gradient) else x
output g1_1_y eval when gps_condition with gradient_g1_1 * g1_1_x + y_intercept_g1_1
output g1_1_x_deg eval when gps_condition with g1_1_x * 180.0 / 3.14159265359
output g1_1_y_deg eval when gps_condition with g1_1_y * 180.0 / 3.14159265359
output g1_1_orientation eval when gps_condition with ( (if abs(g1_1_x - x.last(or: x)) < c_epsilon then false else g1_1_x -  x.last(or: x) < 0.0 ) = orientation_x ) and ( (if abs(g1_1_y -  y.last(or: y)) < c_epsilon then false else g1_1_y - y.last(or: y) < 0.0 ) = orientation_y)
output g1_1_violated_cnd eval when gps_condition with (!isFnc or g1_1_is_intersecting) and g1_1_orientation and (g1_1_x < 55.16368257693344 and g1_1_x > 37.16234411824784)and (g1_1_y < 76.79656929555932 and g1_1_y > 37.16234411824784)
output g1_1_is_violated eval when gps_condition with g1_1_violated_cnd and (sqrt((x.last(or: x) - g1_1_x)**2.0 + (y.last(or: y) - g1_1_y)**2.0) <= dstToPnt)
output g1_1_distance_to eval when gps_condition with if g1_1_violated_cnd then 6371.0 * 2.0 * arcsin(sqrt((sin((x - g1_1_x) / 2.0) ** 2.0) + cos(x) * cos(g1_1_x) * (sin((y - g1_1_y) / 2.0) ** 2.0))) * 1000.0 else 1000.0
output g1_1_time_to eval when gps_condition with g1_1_distance_to / velocity_xy
trigger if gps_condition then g1_1_is_violated.hold(or: false) else false "Violation: g1_1!"

// Polygonline g2: (87.35977750932717, 64.42663705638026) to (60.602330928870806, 44.93835038321772)
constant gradient_g1_2 :Float64 := 1.373001486954943
constant y_intercept_g1_2 :Float64 := -55.51846736397893
output g1_2_is_intersecting eval when gps_condition with gradient != gradient_g1_2 or y_intercept = y_intercept_g1_2
output g1_2_x eval when gps_condition with if isFnc and g1_2_is_intersecting then (y_intercept - y_intercept_g1_2) / (gradient_g1_2 - gradient) else x
output g1_2_y eval when gps_condition with gradient_g1_2 * g1_2_x + y_intercept_g1_2
output g1_2_x_deg eval when gps_condition with g1_2_x * 180.0 / 3.14159265359
output g1_2_y_deg eval when gps_condition with g1_2_y * 180.0 / 3.14159265359
output g1_2_orientation eval when gps_condition with ( (if abs(g1_2_x - x.last(or: x)) < c_epsilon then false else g1_2_x -  x.last(or: x) < 0.0 ) = orientation_x ) and ( (if abs(g1_2_y -  y.last(or: y)) < c_epsilon then false else g1_2_y - y.last(or: y) < 0.0 ) = orientation_y)
output g1_2_violated_cnd eval when gps_condition with (!isFnc or g1_2_is_intersecting) and g1_2_orientation and (g1_2_x < 60.602330928870806 and g1_2_x > 64.42663705638026)and (g1_2_y < 44.93835038321772 and g1_2_y > 64.42663705638026)
output g1_2_is_violated eval when gps_condition with g1_2_violated_cnd and (sqrt((x.last(or: x) - g1_2_x)**2.0 + (y.last(or: y) - g1_2_y)**2.0) <= dstToPnt)
output g1_2_distance_to eval when gps_condition with if g1_2_violated_cnd then 6371.0 * 2.0 * arcsin(sqrt((sin((x - g1_2_x) / 2.0) ** 2.0) + cos(x) * cos(g1_2_x) * (sin((y - g1_2_y) / 2.0) ** 2.0))) * 1000.0 else 1000.0
output g1_2_time_to eval when gps_condition with g1_2_distance_to / velocity_xy
trigger if gps_condition then g1_2_is_violated.hold(or: false) else false "Violation: g1_2!"

// Polygonline g3: (0.8691110136063163, 30.889711873136804) to (39.335291068837705, 6.808748494096017)
constant gradient_g1_3 :Float64 := -1.597368819916523
constant y_intercept_g1_3 :Float64 := 32.278002707317576
output g1_3_is_intersecting eval when gps_condition with gradient != gradient_g1_3 or y_intercept = y_intercept_g1_3
output g1_3_x eval when gps_condition with if isFnc and g1_3_is_intersecting then (y_intercept - y_intercept_g1_3) / (gradient_g1_3 - gradient) else x
output g1_3_y eval when gps_condition with gradient_g1_3 * g1_3_x + y_intercept_g1_3
output g1_3_x_deg eval when gps_condition with g1_3_x * 180.0 / 3.14159265359
output g1_3_y_deg eval when gps_condition with g1_3_y * 180.0 / 3.14159265359
output g1_3_orientation eval when gps_condition with ( (if abs(g1_3_x - x.last(or: x)) < c_epsilon then false else g1_3_x -  x.last(or: x) < 0.0 ) = orientation_x ) and ( (if abs(g1_3_y -  y.last(or: y)) < c_epsilon then false else g1_3_y - y.last(or: y) < 0.0 ) = orientation_y)
output g1_3_violated_cnd eval when gps_condition with (!isFnc or g1_3_is_intersecting) and g1_3_orientation and (g1_3_x < 39.335291068837705 and g1_3_x > 30.889711873136804)and (g1_3_y < 6.808748494096017 and g1_3_y > 30.889711873136804)
output g1_3_is_violated eval when gps_condition with g1_3_violated_cnd and (sqrt((x.last(or: x) - g1_3_x)**2.0 + (y.last(or: y) - g1_3_y)**2.0) <= dstToPnt)
output g1_3_distance_to eval when gps_condition with if g1_3_violated_cnd then 6371.0 * 2.0 * arcsin(sqrt((sin((x - g1_3_x) / 2.0) ** 2.0) + cos(x) * cos(g1_3_x) * (sin((y - g1_3_y) / 2.0) ** 2.0))) * 1000.0 else 1000.0
output g1_3_time_to eval when gps_condition with g1_3_distance_to / velocity_xy
trigger if gps_condition then g1_3_is_violated.hold(or: false) else false "Violation: g1_3!"
\end{lstlisting}

\subsection{Artificial Benchmarks}
\begin{figure}
	\begin{subfigure}{0.49\textwidth}
		\begin{lstlisting}
			input i1: Bool
			input i2: Bool
			input i3: Bool
			output s1
				eval when i1 && i2 && i3 with s2
			output s2
				eval when i1 && i2 with s3
			output s3
				eval when i1 with i1
		\end{lstlisting}
		\caption{Conjuncts Benchmark}
		\label{fig:bench_conj}
	\end{subfigure}
	\hfill
	\begin{subfigure}{0.49\textwidth}
		\begin{lstlisting}
			input bench: Int
			output s1(p1: Int, p2: Int, p3: Int)
				spawn with (bench, bench, bench)
				eval with s2(p1, p2)
			output s2(p1: Int, p2: Int)
				spawn with (bench, bench)
				eval with s3(p1)
			output s3(p1: Int)
				spawn with (bench)
			eval with bench
		\end{lstlisting}
		\caption{Parameter Benchmark}
		\label{fig:bench_param}
	\end{subfigure}
	\begin{center}
	\begin{subfigure}{0.49\textwidth}
		\begin{lstlisting}
			input bench: Int
			output s1 := s2
			output s2 := s3
			output s3 := bench
		\end{lstlisting}
		\caption{Synchronous Access Benchmark}
		\label{fig:bench_sync}
	\end{subfigure}
	\end{center}
	\caption{Example of the benchmarking specifications with $n=3$}
	\label{fig:bench}
\end{figure}

\end{document}